\documentclass[12pt]{article}
\usepackage{amsmath}
\usepackage{amssymb,amsthm}
\usepackage{color}

\newtheorem{Thm}{Theorem}[section]
\newtheorem{theorem}[Thm]{Theorem}

\newtheorem{corollary}[Thm]{Corollary}

\newtheorem{definition}[Thm]{Definition}

\title{Perturbative versus non-perturbative quantum field theory: Tao's method, the Casimir effect and interacting Wightman theories}     
 
\author{Walter Felipe Wreszinski, Instituto de F\'{\i}sica, Universidade de S\~{a}o \\
  Paulo, 05508-090 S\~{a}o Paulo, SP, Brasil, wreszins@gmail.com}

\begin{document}

\maketitle

\begin{abstract}
 We dwell upon certain points concerning the meaning of quantum field theory: the problems with the 
perturbative approach, and the question raised by 't Hooft of the existence of the theory in a well-defined
(rigorous) mathematical sense, as well as some of the few existent mathematically precise results on fully
quantized field theories. Emphasis is brought on how the mathematical contributions help to elucidate or illuminate
certain conceptual aspects of the theory when applied to real physical phenomena, in particular, the singular nature
of quantum fields. In a first part, we present a comprehensive review of divergent versus asymptotic series, 
with qed as background example, as well as a method due to Terence Tao which conveys mathematical sense to divergent series. 
In a second part we apply Tao's method to the Casimir effect in its simplest form, consisting of perfectly conducting 
parallel plates, arguing that the usual theory, which makes use of the Euler-MacLaurin formula, still contains a residual 
infinity, which is eliminated in our approach. In the third part, we revisit the general theory of nonperturbative quantum fields, 
in the form of newly proposed (with Christian Jaekel) Wightman axioms for interacting
field theories, with applications to ``dressed'' electrons in a theory with massless particles (such as qed), as well as
unstable particles. Various problems (mostly open) are finally discussed in connection with concrete models.

\end{abstract}

\vspace{2pc}

\section{Introduction }

One of the most impressive measurements in physics is of the
electron anomalous magnetic moment of the electron, which is known 
with a precision of a few parts in $10^{14}$~\cite{pdg}. On the other hand, 
the theoretical prediction needed to match this result requires the use of
perturbation theory containing up to five loop contributions. As is
well known, higher order contributions in systems
containing an infinite number of degrees of freedom lead to the
appearance of divergences which are taken care of by 
renormalization theory \cite{Weinb1}. The \emph{renormalized}
perturbation series in enormously successful theories such as quantum
electrodynamics (qed) is, however, strongly conjectured to be a \emph{divergent} series
(see \cite{Thi} for the first result in this direction), and G. 't Hooft \cite{Hoo}
has, among others, raised the question whether qed has a meaning in a
rigorous mathematical sense. 
 
Several years ago, Arthur S. Wightman, one of the leaders of axiomatic quantum field theory,
asked the related question ``Should we believe in quantum field theory?'', in a still very
readable review article \cite{Wig1}. After his program of constructing interacting 
quantum field theories in four space-time dimensions failed, the question remains in
the air, and we should like to suggest in the present review that the answer to
it does remain affirmative, although the open problems are very difficult.

One of the problems with quantum field theory (qft) is that it seems to have 
lost contact with its main object of study, viz. explaining the observed phenomena
in the theory of elementary particles. One instance of this fact is that, except for a few 
of the lightest particles, all the remaining ones are unstable, and there exists up
to the present time no single mathematically rigorous model of an unstable particle
(this is reviewed in \cite{Wre1}, and we come back to this point in sections 4 and 5).
Alternative theories have not supplied any new experimentally measured numbers in 
elementary particle physics, such as the impressive one in qed mentioned above. 

A recent, important field of applications of qft came from condensed matter physics: the structure of graphene, 
of which certain (possibly experimentally verifiable) phenomena have been successfully studied by Gavrilov, Gitman et al 
\cite{GavGit} with nonperturbative methods from the theory of strong-field qed with unstable vacuum.
So far, however, the only experimental consequence of a non-perturbative quantum field theoretic model 
concerns the Thirring model \cite{Thirrm} in its
lattice version, the Luttinger model. The latter's first correct solution was due to Lieb and Mattis \cite{LiMa};
the rigorous explanation of the thereby new emergent quasi-particles through fluctuation observables was provided by
Verbeure and Zagrebnov in \cite{VerZag}. Luttinger's model yields a well-established picture of conductivity along 
one-dimensional quantum wires \cite{MaMa}. 

In section 2, we present a comprehensive review of divergent versus asymptotic series, 
with qed as background example, as well as a method due to Terence Tao which conveys mathematical sense to divergent series. 

In section 3, we apply Tao's method to one of the very few nonperturbative effects in qed, the 
Casimir effect in its simplest form, consisting of perfectly conducting 
parallel plates, arguing that the usual theory, which makes use of the Euler-MacLaurin formula, still contains a residual 
infinity, which is eliminated in our approach. The fact that Tao's smoothing of the discrete sums eliminates the divergences 
is seen to be directly related to the singular nature of the quantum fields. 

In section 4, we revisit the general theory of nonperturbative quantum fields, 
in the form of newly proposed (together with Christian Jaekel) Wightman axioms for interacting
field theories, with applications to ``dressed'' electrons in a theory with massless particles (such as qed), as well as
unstable particles \cite{JaWre1}. Here, again, the singular nature of quantum fields appears, in the form of a 
``singularity hypothesis'', characteristic of \emph{interacting} quantum fields, which
plays a major role. It clarifies the role of the nonperturbative wave-function renormalization constant $Z$, and permits
to show that the condition $Z=0$ is not a \emph{universal} one for interacting theories, but is, rather, related to
different phenomena: the ``dressing'' of particles in charged sectors, in the present of massless particles, considered
by Buchholz \cite{Buch}, as well as the existence of unstable particles. As remarked by Weinberg \cite{Weinb1}, the 
characterization of unstable particles by the condition $Z=0$ is an intrinsically \emph{nonperturbative} phenomenon.

Section 5 discusses the latter problem in connection with some concrete models of the Lee-Friedrichs type (see \cite{Wre1}
and references given there), both in the bound state case (atomic resonances) as well as for particles, in which case
the problem remains open. The relation to Haag's theorem \cite{StreWight} is pointed out, completing the ``singularity picture''
of quantum fields developped in the preceding sections. 

Section 6 is reserved to a summary and alternative approaches. 

Section 7 is a conclusion.

\section{Asymptotic versus divergent series: Tao's method}

We now come back to  the perturbation series
for the renormalized gyromagnetic ratio $g$ of the electron (see \cite{JJS}, p. 79,
109), which is a prototype of divergent series in physics:
\begin{equation}
 \frac{g-2}{2} = \frac{1}{2}\left(\frac{\alpha}{\pi}\right)-0.328
\left(\frac{\alpha}{\pi}\right)^{2}+O((\alpha)^{3})  \; ,
\label{(1)}
\end{equation}
where $\alpha = \frac{1}{137.0\cdots}$ denotes the fine-structure
constant. As remarked by Wightman \cite{Wig1}, which we now follow,
the series is likely to be an \emph{asymptotic series}, defined as
below.

In mathematics, divergent series have been studied since the XVIIIth
century and many summability methods have been developed; see for
instance \cite{H}. In particular, Terence Tao \cite{Tao} suggested a 
powerful smoothing method which explains the numbers obtained by these
summability methods. On the other hand, one of the few nonperturbative 
effects in qed, in which a divergent series appears, is the calculation of the 
force between two infinite conducting
planes, also known as the Casimir effect \cite{CasOrig}.  Here, the
zero point energy of the electromagnetic field diverges due to the
existence of an infinite number of normal modes.
In section 3 we study the Casimir effect using Tao's summation method. 

To set the stage, we now recapitulate the
concepts of asymptotic and divergent series, reserving subsection 2.2 to the
presentation of inconsistencies in the treatment of divergent series
by the Euler and Ramanujan summation methods, and subsection 2.3 to a
brief review of Tao's method of smoothed sums.

\subsection{Asymptotic and divergent series}

Given a function $f$, defined on an open interval $a<x<b$,
one says that a function $f_{N}$ defined on the interval is
\emph{asymptotic to f at order $N$ at a} if
\begin{equation}
\lim_{x\to a+} \frac{f(x)-f_{N}(x)}{(x-a)^{N}} = 0
\label{(2)}
\end{equation}
where $x \to a+$ means that $x$ tends to $a$ from the right. This general definition applies in particular to the expression 
$$
f_{N}(x)= \sum_{n=0}^{N}a_{n}(x-a)^{n}
$$
as follows: given a sequence $(a_{0},a_{1},a_{2}, \cdots)$, if
$\sum_{n=0}^{N}a_{n}(x-a)^{n}$ is asymptotic to $f$ of order $N$ at
$a$ for all $N=0,1,2, \cdots$, then one says that the \emph{series}
$\sum_{n=0}^{\infty} a_{n}(x-a)^{n}$ is \emph{asymptotic to $f$ at a}
and writes
\begin{equation}
f(x) \sim \sum_{n=0}^{\infty} a_{n} (x-a)^{n} \;.
\label{(3)}
\end{equation}

Now, if $\lim_{x\to a+}(f(x)-f_{N}(x))(x-a)^{-N}=0$, then it follows
that $\lim_{x\to a}(f(x)-f_{N}(x))(x-a)^{-k}=0$ for
$k=0,1, \cdots N-1$. Therefore, when $\sum_{n=0}^{N}a_{n}(x-a)^{n}$ is
asymptotic to $f$ of order $N$ at $a$,
\begin{eqnarray}
  &&\lim_{x\to a+} f(x) = a_{0} = f(a+)
     \nonumber
  \\
  &&\lim_{x\to a+} (f(x)-f(a+))(x-a)^{-1} = a_{1} = f^{(1)}(a+)
  \nonumber\\
&&\lim_{x\to a+} (f(x)-f(a+)-f^{'}(a+)(x-a))(x-a)^{-2} = a_{2} =
   \frac{f^{(2)}(a+)}{2!} 
  \cdots \label{(4alt)}
  \\
&&\lim_{x\to a+} (f(x)-\sum_{n=0}^{N-1} \frac{f^{(n)}(a+)}{n!}(x-a)^{-N} = a_{N} = \frac{f^{(N)}(a+)}{N!}
\nonumber
\end{eqnarray}
Above, and throughout the paper, the superscripts denote the orders of
the derivatives.  Thus, when $\sum_{n=0}^{N}a_{n}(x-a)^{n}$ is
asymptotic to $f$ of order $N$ at $a$, the function $f$ has
necessarily $N$ derivatives from the right at $a$ and the coefficients
$a_{n}$ are uniquely determined as Taylor coefficients
$$
a_{n} = \frac{f^{(n)}(a+)}{n!} \mbox{ for } n=0,1, \cdots N \;.
$$
Thus, as Wightman observes, the assertion, that the gyromagnetic
anomaly $\frac{g-2}{2}$ has the asymptotic series
$\sum_{n=1}^{\infty}a_{n}(\alpha/\pi)^{n}$ at $0$, means ``no more and
no less'' than the fact that $\frac{g-2}{2}$ is defined for $\alpha$
in some interval $0<\alpha<\alpha_{0}$, and has derivatives of all
orders from the right at $0$: then the above formula for the $a_{n}$
holds for all $N$ with $f=\frac{g-2}{2}$ and $a=0$. If we thus write,
following \eqref{(2)} to \eqref{(4alt)}, 
%
\begin{equation}
f(x) = \sum_{n=0}^{N}\frac{f^{(n)}(a+)}{n!}(x-a)^{n} + R_{N}(x) \;,
\label{(4)}
\end{equation}
where $R_{N}$ is the rest in Taylor's series (see, e.g., \cite{Bu},
p. 126):
\begin{equation}
R_{N}(x) = \int_{a}^{x} \frac{f^{(N+1)}(t)}{N!}(x-t)^{N} dt
\label{(5)}
\end{equation}
and
\begin{equation}
R_{N}(x) = O((x-a)^{N+1}) \;.
\label{(6)}
\end{equation}
Above, $g(x) = O(h(x))$ as $x \to a+$ means that there exists an open
interval $(a,a_{0})$ and a positive number $A$ independent of
$x \in [a,a_{0})$ such that
\begin{equation}
\left |\frac{g(x)}{h(x)} \right| \le A \mbox{ for all } x \in
[a,a_{0}) \;.
\label{(7)}
\end{equation}
It follows then, by what was said before, that
$\sum_{n=0}^{N}\frac{f^{(n)}(a+)}{n!}(x-a)^{n}$ is asymptotic to $f$
at $a+$.


The above is an example of a power series: it may, or not, be
convergent. We come now to a general asymptotic series, which will be
shown to be divergent and will play an important role in section 3,
the \emph{Euler-Maclaurin sum formula} (\cite{H}, Chap. XIII,
p.318). Let a function $f$ be given which satisfies certain conditions
\cite{H}: we shall assume that f is infinitely differentiable to the
right at $a$. Let
\begin{equation}
g(N) \equiv \sum_{n=1}^{N}f(n)- \int_{a}^{N} f(x)dx - C - \frac{1}{2}f(N)
\label{(10.1)}
\end{equation}
Then
\begin{equation}
g(N) \thicksim  \sum_{r=1}^{\infty}(-1)^{r-1}\frac{B_{r}}{2r!}f^{(2r-1)}(N)
\label{(10.2)}
\end{equation}
Above, $C$ is a constant, and the $\{B_{r}\}$ are the Bernoulli numbers, defined recursively by the formula
\begin{equation}
\sum_{k=0}^{s-1} \frac{s!}{k!(s-k)!} B_{k} = s
\label{(10.3)}
\end{equation}
for all $s=1,2,\cdots$, or equivalently, by the generating function
\begin{equation}
\sum_{k=0}^{\infty} \frac{B_{k}}{k!} t^{k} = \frac{t\exp(t)}{(\exp(t)-1)}
\label{(10.4)}
\end{equation}
The first few Bernoulli numbers are
\begin{equation}
B_{0}=1;B_{1}=\frac{1}{2};B_{2}=\frac{1}{6};B_{3}=0;B_{4}=-\frac{1}{30};\cdots
\label{(10.5)}
\end{equation}
We refer to (\cite{Tao}, p. 7 or (\cite{H}, p.320)).

As in the case of power series, the sign $\sim$ in \eqref{(10.2)} means
\begin{equation}
g(n) = \sum_{r=1}^{N} (-1)^{r-1} \frac{B_{r}}{(2r)!} f^{(2r-1)}(n) + R_{N}(n)
\label{(10.6)}
\end{equation}
where, for each finite $N$,
\begin{equation}
\lim_{n\to \infty} \frac{R_{N}(n)}{\sum_{r=1}^{N} (-1)^{r-1} \frac{B_{r}}{(2r)!} f^{(2r-1)}(n)} = 0    
\label{(10.7)}
\end{equation}
but, for each finite $n$,
\begin{equation}
\lim_{N \to \infty} R_{N}(n) = \infty
\label{(10.8)}
\end{equation}
\eqref{(10.8)} expresses the fact that, in general, the Euler-Maclaurin
series \emph{diverges}, due to the fast increase of the Bernoulli
numbers $B_{r}$ with $r$. For an example, take $a=1$ and $f(x)=\log x$
in \eqref{(10.1)}. Then
\begin{equation}
g(n) = \log(n!)-(n+\frac{1}{2})\log n + n + \frac{1}{2}\log(2\pi)
\label{(10.9)}
\end{equation}
with $C=\frac{1}{2}\log(2\pi)$ (for a justification, see \cite{H}). In this case, \eqref{(10.6)} may be written
\begin{equation}
g(n) = \sum_{m=1}^{N} \frac{B_{2m}}{2m(2m-1)n^{2m-1}} + R_{N}(n)
\label{(10.10)}
\end{equation}
where the remainder
\begin{equation}
|R_{N}(n)| \le \frac{B_{2N+2}}{(2N+1)(2N+2)n^{2N+1}}
\label{(10.11)}
\end{equation}
(see (\cite{AS}, 6.1.42)). Therefore, both \eqref{(10.7)} and
\eqref{(10.8)} are seen to hold. The zeroth term in \eqref{(10.9)},
which corresponds to set $g(n)=0$ in \eqref{(10.9)}, is the well-known
Stirling approximation, widely used in statistical mechanics.

The divergence of the series \eqref{(10.1)} is due to the inequalities (\cite{AS},23.1.15):
\begin{equation}
(-1)^{n+1} B_{2n} > \frac{2(2n)!}{(2\pi)^{2n}}
\label{(10.12)}
\end{equation}
as well as the fact that, even if the function $f$ (as a function on the positive reals) is analytic in a neighbourhood of the origin, and thus satisfies the bound
\begin{equation}
|f^{(r)}| \le r!c^{r}
\label{(10.13)}
\end{equation}
for some constant $c$, by Cauchy's integral formula, the $r$-th term
of the series is of order $r!$ by \eqref{(10.12)}. This is explicitly
seen in the special case $f(x)=\log x$ and $a=1$ : \eqref{(10.11)} is
obtained.

We now come back to the power series example \eqref{(1)}. As remarked in the introduction, the first
proof of divergence of perturbation theory for a quantum field
theoretic model (scalar field $\Phi$, with interaction term
$\lambda \Phi^{3}$) was given by Thirring \cite{Thi}. His results are consistent
with the assumption (here extrapolated to the unknown case of qed) that the sum
\begin{equation}
f(\alpha) = \frac{g-2}{2} = \sum_{n=0}^{\infty} a_{n} \alpha^{n}
\label{(8)}
\end{equation}
in \eqref{(1)} is such that the $a_{n}$ increase no worse than $n!$, that is, the $R_{N}(\alpha)$, defined by \eqref{(5)}, is
\begin{equation}
R_{N}(\alpha) = O(N!\alpha^{N})
\label{(9)}
\end{equation}
One may try to find $N$ such that this remainder after $N$ terms is
smallest possible. By \eqref{(9)} and Stirling's formula
(\eqref{(10.9)}, with g(n)=0), and treating $N$ as continuous parameter
for simplicity, we find
\begin{equation}
\log(R_{N}(\alpha)) \approx N\log(N)-N+N\log(\alpha)
\label{(10)}
\end{equation}
from which
\begin{equation}
\frac{d}{dN} \log(R_{N}(\alpha)) \approx \log(N)+\log(\alpha)=0 \mbox{ with solution } \alpha=\frac{1}{N}
\label{(11)}
\end{equation}
The above value of $\alpha$ characterizes a minimum because 
\begin{equation}
\frac{d^{2}}{dN^{2}}(\log(R_{N}(\alpha))_{\alpha=\frac{1}{N}}=\frac{1}{N}>0
\label{(12)}
\end{equation}
We now come to a fundamental question posed by Wightman in
(\cite{Wig1}, p. 994): if a series is asymptotic, is it useful
theoretically or experimentally? His answer is

\emph{Observation A} ``If one knows nothing about the function to which the series is asymptotic, the answer is no''

We shall come to Observation A again later, but, for the moment, adopt
the ``practical'' attitude he advocates in (\cite{Wig1}, p. 995):
according to \eqref{(10)}-\eqref{(12)}, the description of the
function $f(\alpha)$ in \eqref{(8)} should improve with the order of
approximation, until one gets to order $N \approx 137$ by
\eqref{(11)}, and then one should stop; by \eqref{(12)}, beyond that
order, the approximations will become worse. Richard P. Feynman
comments on this issue in the discussion in \cite{Wig2}, p. 226: ``The
question is whether this theory if carried out to the ultimate in all
orders will give a satisfactory series (I don't mean in agreement with
experiments, but with logic). Is it unitary, for example, in the 137th
order? I do not know, and am not at all convinced that it is''. Gerard
't Hooft \cite{Hoo} also observes (p.11): ``Fact is, however, that
there is no proof of the existence of such a model beyond its
perturbation expansion''. These remarks suggest that there are,
indeed, deep issues associated to the (expected) divergence of series
such as \eqref{(1)}.

We now come back to Observation A: what can be said about the function
$f$ to which a given series is asymptotic? One first important remark
concerns \emph{unicity}. The function
\begin{equation}
f_{0}(z) = \exp\left(-\frac{1}{z^{\beta}}\right) \mbox{ with } 0<\beta<1
\label{(13)}
\end{equation}
has the property
\begin{equation}
f_{0}^{(n)}(0+) = 0 \mbox{ for all } n=0,1,2, \cdots
\label{(14)}
\end{equation}
There is, thus, in general, no unique function $f$ to which a given
series is asymptotic, because the function $f_{0}$ may be added to
$f$, by \eqref{(14)}.  There are, however, so-called \emph{summability
  methods}: operations on infinite series, divergent or convergent,
which yield convergent series or functions. For functions, one such
powerful method of summability is that of Borel. Applied to a formal
power series
$$
\sum_{n=0}^{\infty} a_{n}z^{n}
$$
it yields
$$
\frac{1}{x}\int_{0}^{\infty} \exp\left(-\frac{z}{x}\right)
\left(\sum_{n=0}^{\infty}\frac{a_{n}z^{n}}{n!}\right) dz
$$
The above sum under the integral sign is understood as defined by
analytic continuation if necessary. As reviewed by Wightman in
\cite{Wig1}, there are properties which may be assumed on the function
$f$ in \eqref{(3)} which guarantee, for instance, that the Borel
summability method yields a unique answer, the right one. They
typically exclude functions of the form \eqref{(13)}, but are very
difficult to prove in concrete situations such as \eqref{(1)},
because, as remarked by 't Hooft, in the case of qed nothing is known
beyond the asymptotic series on the right of \eqref{(3)} (or, in a
concrete case, \eqref{(1)}).

One may also be concerned with associating a \emph{number} (not a
function) to a given divergent series of scalars by a given
summability method. A unicity issue of a different kind arises: given
one such series, does the possibility exist that different summability
methods yield different (finite) numbers?  This issue is, of course,
not new, and is dealt at length, and with elegance, in G. H. Hardy's
monograph on divergent series \cite{H}. In \cite{H}, p.346, paragraph
13.17, Hardy observes that the summability methods of Euler and
Ramanujan ``have a narrow range and demand great caution in their
application''. Today, it is well known that the usual rules of
calculation cannot be used when handling divergent series (see, e.g.,
the introduction in \cite{Can}). We review the subject, however, in
subsection \ref{sec:2.2}, as a useful introduction to Tao's method in subsection
\ref{sec:2.3}.

\subsection{Inconsistencies in the standard treatment of divergent series}
\label{sec:2.2}

We now consider two examples, the divergent series $S_{0}$ and $S_{1}$:
\begin{equation}
S_{0} = 1+1+1+ \cdots
\label{(2.1)}
\end{equation}
and
\begin{equation}
S_{1} = 1+2+3+ \cdots
\label{(2.2)}
\end{equation}

Although divergent, the series $S_{0}$ and $S_{1}$ may be ``evaluated''
by certain summability methods, two of which we now explain. For the
first, consider the Riemann zeta function $\zeta(s)$, defined as the
series
\begin{equation}
\zeta(s) = \sum_{n=1}^{\infty} \frac{1}{n^{s}}
\label{(2.3)}
\end{equation}
which converges for $\Re(s)> 1$. As shown in \cite{H}, analytic
continuation from a certain complex integral representation for
$\zeta(s)$ may be used to extend $\zeta$ to values of $s$ beyond
$\Re(s)>1$. Let $\zeta_{a.c.}$ denote the analytically continued
function, and call \emph{Ramanujan sums}, denoted by the symbol
${\cal R}$, as in \cite{H}, the corresponding sums, e.g., for
\eqref{(2.1)} and \eqref{(2.2)},
\begin{equation}
S_{0}({\cal R}) = \zeta_{a.c.}(0) = -\frac{1}{2} 
\label{(2.4)}
\end{equation}
and
\begin{equation}
S_{1}({\cal R}) = \zeta_{a.c.}(-1) = -\frac{1}{12} \; .
\label{(2.5)}
\end{equation}

For the second method, due to Euler, consider a (possibly divergent) series $S \equiv \sum_{n=1}^{\infty} a_{n}$, and define the function
\begin{equation}
f_{S}(t) = \sum_{n=1}^{\infty} t^{n}a_{n}
\label{(2.6)}
\end{equation}
for $|t| \le 1$. Assume that $\sum_{n=1}^{\infty} t^{n}a_{n} < \infty$ if $|t|<1$, and define the Euler sum of the original series $S$ by
\begin{equation}
S({\cal E}) = \lim_{t\to 1-} \sum_{n=1}^{\infty} t^{n}a_{n}
\label{(2.7)}
\end{equation}

It is clear that
\begin{equation}
f_{S_{0}}(t) = \frac{1}{1-t}
\label{(2.8.1)}
\end{equation}
and
\begin{equation}
f_{S_{1}}(t) = \frac{1}{(1-t)^{2}}
\label{(2.8.2)}
\end{equation}
and thus, by \eqref{(2.7)},
\begin{equation}
S_{0}({\cal E}) = S_{1}({\cal E}) = +\infty \;.
\label{(2.8.3)}
\end{equation}
Thus, the Euler sums of $S_{0}$ and $S_{1}$ disagree with their
Ramanujan counterparts, as seen from \eqref{(2.4)}, \eqref{(2.5)} and
\eqref{(2.8.3)}.  We have, however, for $\Re( s)>1$,
\begin{equation}
1^{-s}+3^{-s}+5^{-s}+\cdots = (1-2^{-s})\zeta(s) 
\label{(2.9)}
\end{equation}
\begin{equation}
2^{-s}+4^{-s}+\cdots = 2^{-s}\zeta(s)
\label{(2.10)}
\end{equation}
Subtracting \eqref{(2.10)} from \eqref{(2.9)} we find
\begin{equation}
(1-2^{1-s})\zeta(s) = 1^{-s}-2^{-s}+3^{-s}-4^{-s}+\cdots
\label{(2.11)}
\end{equation}
or
\begin{equation}
\zeta(s) = (1-2^{1-s})^{-1}(1^{-s}-2^{-s}+3^{-s}-4^{-s}+\cdots)
\label{(2.12)}
\end{equation}
\eqref{(2.12)} expresses $\zeta(s)$ as an alternating series, to
which, when Euler summation is applied, indeed yields (as may be
easily checked):
\begin{equation}
S_{0}({\cal E}) = -\frac{1}{2}
\label{(2.13)}
\end{equation}
as well as
\begin{equation}
S_{1}({\cal E}) = -\frac{1}{12}
\label{(2.14)}
\end{equation}
What is, now, the Ramanujan sum associated to $\zeta$, given by the
r.h.s. of \eqref{(2.12)}, in the case $s=-1$?  For $s=-1$, it
corresponds to the divergent alternating series
\begin{equation}
S_{1}^{'} \equiv -\frac{1}{3}(1-2+3-4+\cdots)
\label{(2.15)}
\end{equation}
We find, however, that operations of the same type as those leading to
\eqref{(2.12)} (summation term by term, multiplication by a scalar)
yield the following Ramanujan sum values:
\begin{equation}
2+4+6+\cdots = 2(1+2+3+\cdots) = 2S_{1}({\cal R}) = -\frac{1}{6}
\label{(2.16)}
\end{equation}
and further
\begin{equation}
1+3+5+\cdots = 2+4+6+\cdots -(1+1+1+\cdots)=2S_{1}({\cal R})-S_{0}({\cal R})=-\frac{1}{6}+\frac{1}{2}=\frac{1}{3}
\label{(2.17)}
\end{equation}
by \eqref{(2.4)} and \eqref{(2.5)}; thus, by \eqref{(2.15)}
\begin{equation}
-3S_{1}({\cal R})=(1+3+\cdots)-(2+4+\cdots)=\frac{1}{3}+\frac{1}{6}=\frac{1}{2}
\label{(2.18)}
\end{equation}
and finally
\begin{equation}
S_{1}^{'}({\cal R})=\frac{-1}{3}\times \frac{1}{2}=-\frac{1}{6} \neq S_{1}({\cal R})
\label{(2.19)}
\end{equation}
Therefore, although \eqref{(2.12)} coincides with \eqref{(2.3)}
throughout the region of convergence, its Ramanujan sum \emph{differs}
from the corresponding one for \eqref{(2.3)}. This means that the
algebraic operations of scalar multiplication and term-by-term
summation are not preserved by analytic continuation, about which we
shall shortly say more. Incidentally, \eqref{(2.19)} and \eqref{(2.14)}
imply

\emph{Observation B} The Euler and Ramanujan summation methods may yield \emph{different finite} values for a given divergent series.  
 
Observation B does not seem, surprisingly, to have been made before,
although Hardy (\cite{H}, p.345, paragraph 13.17) remarked the
disagreement between \eqref{(2.5)} and the second of \eqref{(2.8.3)}.

There exists a way of reconciling \eqref{(2.15)} with \eqref{(2.5)},
which is pointed out by Hardy (\cite{H}, p.346, paragraph
13.17). Interpret, in \eqref{(2.18)}, $1+3+\cdots$ as
$1+0+3+0+\cdots$, and $2+4+\cdots$ as $0+2+0+4+\cdots$ , which are
consistent with the r.h.s. of \eqref{(2.9)} and \eqref{(2.10)} for
$s=-1$. In this way, one obtains for the r.h.s. of \eqref{(2.15)}
\begin{equation}
S_{1}^{''}({\cal R}) = -\frac{1}{3}((-1)(-\frac{1}{12})-2(\frac{-1}{12}))=-\frac{1}{3}\times \frac{1}{4}=-\frac{1}{12}
\end{equation}
which does agree with \eqref{(2.5)}! The value of
$S_{1}^{''}({\cal R})$ depends, however, on the bizarre properties
$(1+0+3+0+\cdots) \neq (1+3+\cdots)$, as well as
$(0+2+0+4+\cdots) \neq (2+4+\cdots)$, i.e., the sums are not invariant
by the addition of zeroes, or, alternatively, term-by-term summation
is not allowed (or does not hold). Summarizing, analytic continuation
does not preserve the fundamental algebraic properties of scalar
multiplication and term-by-term summation; as a consequence,
Observation B holds. See also \cite{Can} about this issue.

Equation \eqref{(2.12)} motivates the introduction of a third
divergent series, called \emph{Grandi's series} $G$:
\begin{equation}
G = \sum_{n=1}^{\infty} (-1)^{n-1} = 1-1+1-\cdots
\label{(2.20)}
\end{equation}
 
The partial sums
\begin{equation}
\sum_{n=1}^{N} (-1)^{n-1} = \frac{1}{2} + \frac{1}{2}(-1)^{N-1}
\label{(2.21)}
\end{equation}
oscillate between $1$ and $0$, so that the series is neither
conditionally nor absolutely convergent. By \eqref{(2.12)} and
\eqref{(2.4)}, its Ramanujan sum is
\begin{equation}
G({\cal R}) = \frac{1}{2}
\label{(2.22)}
\end{equation}
Parenthetically, define the Ces\`{a}ro (C) sums by, for $S=\sum_{n=1}^{\infty}a_{n}$, 
\begin{equation}
S({\cal C}) = \lim_{n\to\infty}\frac{1}{n+1}\sum_{l=0}^{n}\sum_{k=0}^{l}a_{k}
\label{(2.23)}
\end{equation}
Equation \eqref{(2.23)} corresponds to the process of averaging the partial sums and is, yet, another popular summability method. Using \eqref{(2.21)}, the reader may verify that
\begin{equation}
G({\cal C}) = \frac{1}{2}
\label{(2.24)}
\end{equation}
and it is immediate that
\begin{equation}
G({\cal E}) = \frac{1}{2}
\label{(2.25)}
\end{equation}
where ${\cal E}$ means ``Euler''. Thus, for Grandi's series, all three
summability methods, Ramanujan, Euler and Ces\`{a}ro, yield the same
sum $\frac{1}{2}$. It might be expected, therefore, that this number
$\frac{1}{2}$ is indeed attached to the series $G$ in some way,
particularly because the partial sums oscillate between one and zero
by \eqref{(2.21)}. We shall see how in the next subsection.

\subsection{A new look at divergent series: Terence Tao's method of smoothed sums}
\label{sec:2.3}

The Ramanujan method of analytic continuation, when used in connection
with the zeta function for general $s$ in \eqref{(2.3)}, yields a
generalization of \eqref{(2.13)} and \eqref{(2.14)}, namely
\begin{equation}
S_{s}({\cal R}) = \zeta_{a.c.}(-s) = -\frac{B_{s+1}}{s+1}
\label{(4.1)}
\end{equation}
where the $B_{s}$ are the Bernoulli numbers \eqref{(10.3)}, \eqref{(10.4)}; the corresponding formal series are
\begin{equation}
S_{s} \equiv \sum_{n=1}^{\infty} n^{s}
\label{(4.2)}
\end{equation}
with partial sums
\begin{equation}
S_{s}(N) = \sum_{n=1}^{N} n^{s} = \frac{1}{s+1}N^{s+1} +\frac{1}{2}N^{s}+\frac{s}{12}N^{s-1}+\cdots+B_{s}N
\label{(4.3)}
\end{equation}
(\emph{Faulhaber's formula}). For $s=0,1$, \eqref{(4.1)} yields the Ramanujan sums
\begin{equation}
\sum_{n=1}^{\infty}1 = 1+1+ \cdots = -\frac{1}{2}
\label{(4.4)}
\end{equation}
\begin{equation}
\sum_{n=1}^{\infty} n = 1+2+3+ \cdots = -\frac{1}{12}
\label{(4.5)}
\end{equation}
We complete the list with Grandi's series \eqref{(2.20)}, with partial sums \eqref{(2.21)}, and Ramanujan sum
\begin{equation}
\sum_{n=1}^{\infty} (-1)^{n-1} = 1-1+1- \cdots = \frac{1}{2}
\label{(4.6)}
\end{equation}

As remarked by Terence Tao \cite{Tao}, the partial sums $S_{s}(N)$ do not ressemble the r.h.s. of \eqref{(4.1)}, e.g., for $s=0,1$,
\begin{equation}
S_{0}(N) = \sum_{n=1}^{N} 1 = N
\label{(4.7)}
\end{equation}
\begin{equation}
S_{1}(N) = \sum_{n=1}^{N} n = \frac{1}{2}N^{2} + \frac{1}{2}N
\label{(4.8)}
\end{equation}

Tao's main observation is that, if $N$ is viewed as a real number in
\eqref{(4.3)}, $S_{s}(N)$ has jump discontinuities at each positive
integer $N$, which play a crucial role in its asymptotic behavior with
$N$; he proposes to replace $S_{s}(N)$ by \emph{smoothed sums}

\begin{equation}
S_{s}(\eta_{N}) \equiv \sum_{n=1}^{\infty} \eta\left(\frac{n}{N}\right) n^{s}
\label{(4.9)}
\end{equation}
where 
\begin{equation}
\eta_{N}(x) \equiv \eta\left(\frac{x}{N}\right)
\label{(4.9a)}
\end{equation}
Above, $\eta: \mathbf{R}_{+} \to \mathbf{R}$, the \emph{cutoff function}, is a bounded function of compact support, which we take, without loss of generality, to be the interval $[0,1]$, i.e., $\eta(x)=0 \mbox{ if } x \notin [0,1]$, and such that
\begin{equation}
\eta(0+) = 1
\label{(4.10)}
\end{equation}
The values of $\eta$ on the negative real axis are of no concern. We
assume, for simplicity, that $\eta$ is smooth (infinitely
differentiable) on the set $(0,1]$: this means that the function and
all its derivatives at $x=1$ are zero. The latter behavior contrasts
with that of the characteristic function $\chi$ of $[0,1]$, which is
defined to equal one if $x \in [0,1]$, and zero otherwise. If $\eta$
is replaced by $\chi$, a finite (step) discontinuity at $x=1$ takes
place, and in \eqref{(4.9)} we recover the traditional partial sums
\eqref{(4.3)}.

We now review the simplest case treated by Tao, that of Grandi's
series \eqref{(4.6)}, for which the corresponding smoothed sum
\eqref{(4.9)} may be written by regrouping the summands as
\begin{eqnarray*}
G(\eta_{N}) = \sum_{n=1}^{\infty} \eta\left(\frac{n}{N}\right) (-1)^{n-1} = \\
= \frac{\eta(\frac{1}{N})}{2} + \sum_{n=1}^{\infty} \frac{[\eta(\frac{(2m-1)}{N})-2\eta(\frac{2m}{N})+\eta(\frac{(2m-1)}{N})]}{2}
\end{eqnarray*}
If $\eta$ is twice continuously differentiable, it follows from the
Taylor expansion that each summand in $G(\eta_{N})$ is
$O(\frac{1}{N^{2}})$ (because it is a double difference). From the
compact support of $\eta$ the number of terms in the infinite sum in
$G(\eta_{N})$ is finite and $O(N)$. Thus, the sum occurring in the
expression for $G(\eta_{N})$ is $O(\frac{1}{N})$ and, due to
\eqref{(4.10)}, the first term in the expression of $G(\eta_{N})$ is
$O(\frac{1}{N})$, too. Hence,
\begin{equation}
G(\eta_{N}) = \frac{1}{2} + O\left(\frac{1}{N}\right)
\label{(4.11)}
\end{equation}
The common Ramanujan, Euler and Ces\`{a}ro sum of $G$ - namely, one
half- appears therefore as leading term in the smoothed asymptotics.

The above method may be applied to $S_{s}(\eta_{N})$, defined by \eqref{(4.9)}, yielding \cite{Tao}
\begin{equation}
S_{s}(\eta_{N}) = -\frac{B_{s+1}}{s+1} + C_{\eta_{N},s}N^{s+1} + O\left(\frac{1}{N}\right)
\label{(4.12)}
\end{equation}
for any fixed $s=0,1,2, \cdots$, where
\begin{equation}
C_{\eta,s} \equiv \int_{0}^{\infty} x^{s} \eta(x) dx
\label{(4.13)}
\end{equation}
is the Mellin transform of $\eta$ (see, e.g., \cite{Hoc}, p. 80, Chap. 3). By \eqref{(4.9a)},
\begin{equation}
C_{\eta_{N},s} = C_{\eta,s}N^{s+1}
\label{(4.13a)}
\end{equation}
By \eqref{(4.12)} we obtain, for $s=0$ and $s=1$ the results
\begin{equation}
  S_{0}(\eta_{N}) = \sum_{n=1}^{\infty} \eta\left(\frac{n}{N}\right)
  = -\frac{1}{2} + C_{\eta_{N},0}N + O\left(\frac{1}{N}\right)
\label{(4.14)}
\end{equation}
\begin{equation}
  S_{1}(\eta_{N}) = \sum_{n=1}^{\infty} n \eta\left(\frac{n}{N}\right)
  = -\frac{1}{12} + C_{\eta_{N},1}N^{2} + O\left(\frac{1}{N}\right)
\label{(4.15)}
\end{equation}
Equations \eqref{(4.14)} and \eqref{(4.15)} should be compared with
\eqref{(2.4)} and \eqref{(2.5)}, respectively. They show that
Ramanujan sums $\frac{-1}{2}$ and $\frac{-1}{12}$ appear as the
leading terms of asymptotic expansions in $N$ of $S_{0}(\eta_{N})$ and
$S_{1}(\eta_{N})$, respectively, for large $N$, which, as remarked by
Tao, do not contradict the positivity of the summands on the l.h.s. of
\eqref{(4.14)} and \eqref{(4.15)}, because the $C_{\eta_{N},s}$ terms
are, by \eqref{(4.13a)}, both strictly positive and dominate for large
$N$. In particular, the inconsistencies pointed out in section 2.1
connected with term-by-term summation and scalar multiplication, in
the conventional treatment, disappear when definition \eqref{(4.9)} of
smoothed partial sums is adopted. The simplest example is that of the
scalar multiplication, take the first (innocent-looking) equation of
\eqref{(2.16)}:
$$
2+4+6+ \cdots = 2(1+2+3+ \cdots)
$$
This is no longer true for smoothed  sums. For,
$$
\sum_{n=1}^{\infty} 2n \eta\left(\frac{2n}{N}\right) \neq 2
\sum_{n=1}^{\infty} n \eta\left(\frac{n}{N}\right)
$$
Similarly, term-by-term summation yields no contradictions, because of
the necessary associated shifts in the arguments of the function
$\eta$.

We now come back to \eqref{(4.12)}. This formula was proved by Tao as
a corollary of his very elegant and simple proof of the
Euler-Maclaurin formula. Let $f$ satisfy the following assumption:

\emph{Assumption A} f is a real valued function on $(0,\infty)$, at
least continuously differentiable up to order $(s+2)$, vanishing,
together with all its derivatives up to order $(s+1)$ at $N$.

Let 
\begin{equation}
\tilde{g}(N) \equiv \int_{0}^{N} f(x)dx - \frac{1}{2}f(0) - \sum_{n=1}^{N}f(n)
\label{(4.16)}
\end{equation}
Then,
\begin{equation}
\tilde{g}(N) = \sum_{k=2}^{s+1} \frac{B_{k}}{k!} f^{(k-1)}(0) + O(N||f||_{C^{s+2}})
\label{(4.17)}
\end{equation}
where $s=1,2, \cdots$, and
\begin{equation}
||f||_{C^{k}} \equiv \sup_{x\in \mathbf{R}}|f^{(k)}(x)|
\label{(4.18)}
\end{equation}
Applying \eqref{(4.17)} with 
\begin{equation}
f(x) = x^{s} \eta(\frac{x}{N})
\label{(4.19)}
\end{equation}
we see from \eqref{(4.16)}, the conditions on $\eta$ stated between
equations \eqref{(4.9)} and \eqref{(4.10)}, and \eqref{(4.13a)} that
the l.h.s. of \eqref{(4.17)} equals $C_{\eta,s}N^{s+1}$. It may be
easily checked that all terms in the sum in \eqref{(4.17)} vanish,
except for the $k=s+1$ term, which is $\frac{B_{s+1}}{s+1}$. Applying
the product rule and the chain rule to definition \eqref{(4.18)}, and
using that $\eta$ is supported on the interval $[0,1]$, we also check
that
\begin{equation}
||f||_{C^{s+2}} = O\left(\frac{1}{N^{2}}\right)
\label{(4.20)}
\end{equation}
from which \eqref{(4.12)} follows. From \eqref{(4.12)},
\begin{equation}
\lim_{N \to \infty} [S_{s}(\eta_{N}) - C_{\eta,s} N^{s+1}] = -\frac{B_{s+1}}{s+1}
\label{(4.21)}
\end{equation}
One should compare \eqref{(4.21)} with the result which follows from Faulhaber's formula \eqref{(4.3)}, namely, $S_{s}(N)-\frac{N^{s+1}}{s+1} \to \infty$. This is because when $\eta$ is replaced by the characteristic function of $[0,1]$, Assumption A does not hold, due to the boundary terms at $N$ (see exercise 5 in \cite{Tao}).

\section{Application of Tao's method to the Casimir effect for perfectly conducting parallel plates}
\label{sec:3}

In this section, we demonstrate the applicability of Tao's method to
an important effect in nonperturbative qed (in the sense of being presumably exact, although
this will be reviewed), the Casimir effect
(\cite{IZ}, \cite{Mil}, \cite{Mi}. For a review, see \cite{Plu}, for
recent advances see \cite{BKMM}, and for a pedagogical treatment,
\cite{Far}. We shall follow \cite{MaB}, see also \cite{Mil},
sec. 2.7. Only the very simplest case, that of perfectly conducting
parallel plates, will be considered.

Consider an empty coboidic box $\Lambda$ with perfectly conducting
boundaries. It has thickness $d$ and lateral sizes of surface $L^{2}$:
$$
\Lambda = \{\vec{x}=(x,y,z)|0\le x\le d, -L/2 \le y \le L/2, -L/2 \le z \le L/2 \}
$$
The electric field in the box is the solution of Maxwell's equations, with proper boundary conditions (see \cite{MaB}, and \cite{Mil}, section 2.7). It yields a classical energy, whose quantization provides the expression for the total Hamiltonian $H_{\Lambda}$:
\begin{eqnarray*}
H_{\Lambda} = \sum_{\vec{k},\lambda,'}\hbar \omega_{\vec{k}}a^{\dag}_{\vec{k},\lambda}a_{\vec{k},\lambda}+\\
+\frac{1}{2}\sum_{\vec{k},\lambda,'}\hbar \omega_{\vec{k}}
\end{eqnarray*}
Above, the modes $(\vec{k},\lambda)$ are described by a wave-number
$\vec{k}=(k_{x},k_{y},k_{z})$ and a polarization index $\lambda=1,2$,
$a^{\dag}(\vec{k}),\lambda, a(\vec{k},\lambda)$ are creation and
annihilation operators satisfying commutation relations
$[a(\vec{k},\lambda),a^{\dag}(\vec{k^{'}},\lambda^{'})]=
\delta_{\vec{k},\vec{k^{'}}}\delta_{\lambda,\lambda^{'}}$,
$\omega_{\vec{k}}=c|\vec{k}|$ is the photon energy and the prime in
$(\vec{k}, \lambda, ')$ means that only one polarization is possible
when one of the wave-numbers $(k_{x},k_{y},k_{z})$ equals zero. The
wave-numbers $k_{x},k_{y},k_{z}$ and the eigenfrequencies are
\begin{eqnarray*}
k_{x}=\frac{\pi n_{x}}{d}, k_{y}=\frac{\pi n_{y}}{L}, k_{z}=\frac{\pi n_{z}}{L}, n_{x}, n_{y}, n_{z}=0,1,2, \cdots ;\\
\omega_{\vec{k}} = c|\vec{k}| = c\sqrt{\left(\frac{\pi n_{x}}{d}\right)^{2}+\left(\frac{\pi n_{y}}{L}\right)^{2}+\left(\frac{\pi n_{z}}{L}\right)^{2}}
\end{eqnarray*}
where $c$ denotes the speed of light.
The boundary condition on the electric field, i.e., that its
tangential components vanish on the boundary of $\Lambda$, lead to the
property of its $y$ and $z$ components to be proportional to
$\sin(k_{x}x)$, with $k_{x}$ as above. The (infinite) last term in the
Hamiltonian represents the zero-point energy of the electromagnetic
field. As in \cite{MaB}, we adopt Casimir's point of view that the
vacuum energy
\begin{eqnarray*}
\Sigma_{\Lambda} = (0|\frac{1}{8\pi}\int_{\Lambda} dx (|\vec{E}(\vec{x})|^{2}+|\vec{B}(\vec{x})|^{2})|0)= \\
=\frac{1}{2}\sum_{\vec{k},\lambda,'}\hbar \omega_{\vec{k}} 
\end{eqnarray*}
represents mean square fluctuations of the fields in the box $\Lambda$
that exist even in the absence of photons but may produce physically
observable effects because they depend on the geometry (shape, size)
of the spatial domain containing the field. Accordingly, let
$f(d) \equiv \frac{F(d)}{L^{2}}$ denote the force per unit surface
induced by these vacuum fluctuations between two faces of the metallic
box at distance $d$. Since a real metal is characterized by a
frequency-dependent dielectric function $\epsilon(\omega)$ such that
$\epsilon(\omega) \to \infty$ as $\omega \to 0$, but which tends to
the vacuum value $\epsilon_{0}$ as $\omega \to \infty$, namely, when
$\omega \gg \omega_{a}$, with $\omega_{a}$ a characteristic atomic
frequency, high enough frequencies should not contribute to the
force. For this reason, we introduce a cutoff function
$g(\omega/\omega_{a})$ in the above formula for the zero point energy,
such that
\begin{equation}
g(0)=1
\label{(5.1)}
\end{equation}
as well as
\begin{equation}
g\left(\frac{\omega}{\omega_{a}}\right) \to 0 \mbox{ as } \omega \to \infty
\label{(5.2)}
\end{equation}
This means to consider $\Sigma_{\Lambda}$ above as given by
\begin{equation}
\Sigma_{\Lambda} = \frac{1}{2} \sum_{\vec{k},\lambda,'}\hbar \omega_{\vec{k}} g\left(\frac{\omega_{k}}{\omega_{a}}\right)
\label{(5.3)}
\end{equation}
The cutoff function will be removed at the end by letting $\omega_{a} \to \infty$.

Extending the plates to infinity in the $y,z$ directions yields the energy per unit surface
\begin{eqnarray*}
u_{d} = \lim_{L \to \infty} \frac{\Sigma_{\Lambda}}{L^{2}} = 2 \frac{1}{\pi^{2}}\int_{0}^{\infty}dk_{y}\int_{0}^{\infty} dk_{z}\\
(\frac{1}{2}\sum_{0}^{\infty,'} \hbar \omega_{n}(\sqrt(k_{y}^{2}+k_{z}^{2})))g(\frac{\omega_{n}(\sqrt(k_{y}^{2}+k_{z}^{2}))}{\omega_{a}})=\\
= 2 \frac{1}{4\pi}\sum_{n=0}^{\infty,'}\int_{0}^{\infty} dq q \hbar \omega_{n}(q)g\left(\frac{\omega_{n}(q)}{\omega_{a}}\right)
\end{eqnarray*}
where
$$
\omega_{n}(q) = c\sqrt((\frac{\pi n}{d})^{2}+q^{2})
$$
Above, $\vec{q}$ is the two-dimensional wave-number vector in the
$(y,z)$-plane, $q=|\vec{q}|$, and we have introduced polar coordinates
in the $(y,z)$ plane, with angular sector $\frac{2\pi}{4}$. The
prefactor $2$ in $u_{d}$ is due to the two polarization states of the
photon an the prime in the sum means that the term $n=0$ must have an
additional factor $\frac{1}{2}$. With the change of variable
$q \to \omega=\omega_{n}(q), \omega d\omega= c^{2}q dq$, $u_{d}$ may
be written finally
$$
u_{d} = \frac{1}{4\pi c^{2}}\sum_{n=0}^{\infty,'}\int_{\frac{c\pi n}{d}}^{\infty}d\omega \omega \hbar \omega g(\frac{\omega}{\omega_{a}}) 
$$
and therefore
\begin{equation}
u_{d}=\frac{\hbar c \pi^{2}}{2d^{3}}\sum_{n=0}^{\infty,'}F(n)
\label{(5.4)}
\end{equation}
where
\begin{equation}
F(s) \equiv \int_{s}^{\infty} dv v^{2} g\left(\frac{\pi cv}{d\omega_{a}}\right)
\label{(5.5)}
\end{equation}
A possibility for a cutoff function $g$ satisfying \eqref{(5.1)} and \eqref{(5.2)} is to consider, for each finite integer $N$ (this is
no restriction, because one may always consider for any given sequence of numbers $\alpha_{i}, i=1, \cdots$ the largest 
integers $N_{i}\le \alpha_{i}, i=1, \cdots$) the function $g(x)=g_{N}(x) = \chi(\frac{x}{N})$,
where $\chi$ denotes the characteristic function of the interval $[0,1]$, i.e., $\chi(x)=1$ if $x \in [0,1]$ and $\chi(x)=0$ otherwise.
Recalling, however, the problems mentioned after \eqref{(4.21)}, we are led to adopt the following \emph{definition} for $g$ instead:
\begin{equation}
g_{N}(x) \equiv \eta_{N}(x)
\label{(5.6)}
\end{equation}
with $\eta_{N}$ defined as in \eqref{(4.9a)}, \eqref{(4.10)}.
Since $\frac{\omega_{n}(q)}{\omega_{a}} \le N$ implies $\frac{c \pi n}{d} \le N$, we may write \eqref{(5.4)} and \eqref{(5.5)} in the form
\begin{equation}
u_{d} = \frac{\hbar c\pi^{2}}{2d^{3}}\lim_{N \to \infty} \sum_{n=0}^{N,'} F_{N}(n)
\label{(5.7)}
\end{equation}
where
\begin{equation}
F_{N}(n) \equiv \int_{n}^{\infty} dv v^{2} g_{N}\left(\frac{\pi cv}{d\omega_{a}}\right)
\label{(5.8)}
\end{equation}
The force between the plates, assuming that the electromagnetic field is
entirely enclosed in the cavity, is
$f_{d} = -\frac{\partial}{\partial d}u_{d}$, with $u_{d}$ given by
\eqref{(5.4)} and $F$ replaced by $F_{N}$. There is, however, also a field in the space external
to the cavity: the external face of the plate at $d$ will be subject
to a force in the opposite direction due to the vacuum fluctuations in
the semi-infinite space to its right, namely
\begin{equation}
f_{ext} = -\left(- \lim_{d \to \infty} \frac{\partial}{\partial d}u_{d}\right)
\label{(5.9)}
\end{equation}
which yields
\begin{equation}
u_{d}^{ext} \approx
\frac{d}{2\pi^{2}c^{3}}\int_{0}^{\infty}dv\int_{v}^{\infty}d\omega
\omega (\hbar \omega)
g\left(\frac{\omega}{\omega_{a}}\right)
\label{(5.10)}
\end{equation}
By the formula for $u_{d}$ just preceding \eqref{(5.4)}, the corresponding energy per unit surface is
\begin{equation}
u_{d}^{ext} \approx \frac{\hbar c\pi^{2}}{2d^{3}}\int_{0}^{\infty}ds F_{N}(s)
\label{(5.11)}
\end{equation}
as $d \to \infty$, up to negligible corrections which are
discarded. Note that this Ansatz amounts to normalize the force in
such a way that a single plate in infinite space feels no resulting
force. Define, thus, the total energy $u^{t}$ per unit surface by
\begin{equation}
u^{t} = u_{d} - u_{d}^{ext}
\label{(5.12)}
\end{equation}  
By \eqref{(5.4)} and \eqref{(5.10)}, it is given by
\begin{equation}
u^{t} = \frac{\pi^{2}\hbar c}{2d^{3}}\lim_{N \to \infty}u_{N}^{t}
\label{(5.13a)}
\end{equation}
where
\begin{equation}
u_{N}^{t} \equiv \sum_{n=1}^{\infty}F_{N}(n)+\frac{1}{2}F_{N}(0)-\int_{0}^{\infty}F_{N}(s)
\label{(5.13b)}
\end{equation}
Our main result is the following theorem, which is an application of Tao's method to a less trivial case then \eqref{(4.9)}, i.e., in which the relevant function $f$ is not a polynomial:

\begin{theorem}
\label{thm:2}
\begin{equation}
\label{(5.14)}
u^{t} = \frac{-2\pi^{2}\hbar c}{2d^{3}}\frac{B_{4}}{4!} = \frac{-\pi^{2}\hbar c}{720 d^{3}}
\end{equation}
where $B_{4}$ is a Bernoulli number \eqref{(10.5)}. 
\end{theorem}

\begin{proof}

Let $\lambda \equiv \frac{\pi c}{d\omega_{a}}$ By \eqref{(5.7)},
\begin{eqnarray*}
F_{N}^{(1)}(s) = -s^{2}g_{N}(\lambda s)\\
F_{N}^{(2)}(s) = -2sg_{N}(\lambda s)-\lambda s^{2}g^{(1)}(\lambda s)\\
F_{N}^{(3)}(s) = -2\eta(0)-4s\lambda g_{N}^{(1)}(\lambda s)-\lambda^{2}s^{2}g_{N}^{(2)}(\lambda s)\\
F_{N}^{(4)}(s) = -4\lambda g_{N}^{(1)}(\lambda s)-6s\lambda^{2}g_{N}^{(2)}(\lambda s)-\lambda^{3}s^{2}g_{N}^{(3)}(\lambda s)\\
F_{N}^{(5)}(s) = -10\lambda^{2}g_{N}^{(2)}(\lambda s)-8s\lambda^{3}g_{N}^{(3)}(\lambda s)-\lambda^{4}s^{2}g_{N}^{(4)}(\lambda s)
\end{eqnarray*}

Above, as before, the superscripts denote the order of the derivatives. From the above formulas, we find that $F_{N}^{(1)}(0)=0$, $F_{N}^{(2)}(0)=0$. Inserting these results into \eqref{(4.16)}, \eqref{(4.17)} and \eqref{(4.18)}, and taking into account \eqref{(5.13a)} and \eqref{(5.13b)}, we obtain
\begin{equation}
u_{N}^{t} = -\frac{B_{4}}{4!}F_{N}^{(3)}(0) + O(N||F_{N}||_{C^{5}})
\label{(5.15)}
\end{equation}
where, from the above explicit formula for $F_{N}^{(5)}$,
\begin{equation}
||F_{N}||_{C^{5}} \equiv \sup_{x\in \mathbf{R}_{+}} |F_{N}^{(5)}(x)| \le 10 \lambda^{2} \frac{\lambda}{N^{2}}+O\left(\frac{1}{N^{3}}\right)
\label{(5.16)}
\end{equation} 
The explicit formula for $F_{N}^{(3)}$ yields
\begin{equation}
F_{N}^{3}(0) = -2\eta(0+) 
\label{(5.17)}
\end{equation}
Inserting \eqref{(5.16)} and \eqref{(5.17)} into \eqref{(5.15)}, and taking into account the normalization \eqref{(4.10)}, we obtain
\begin{equation}
u_{N}^{t} = \frac{-\pi^{2}\hbar c}{720 d^{3}} + O\left(\frac{1}{N}\right)
\label{(5.18)}
\end{equation}
from which, together with \eqref{(5.13a)}, \eqref{(5.14)} follows.

\end{proof}

We have shown that Tao's method of smoothed sums yields the correct
formula \eqref{(5.14)} for the total energy density for the simplest
Casimir effect, of perfectly conducting parallel plates. Most proofs
of this effect use the Euler-MacLaurin formula, and write
\eqref{(5.14)} in the form
\begin{equation}
u^{t} = -\frac{\pi^{2}\hbar c}{720 d^{3}} + O(d^{-4})
\label{(6.1)}
\end{equation}
(see, e.g., \cite{Sp}, p.171, (13.121), for a recent reference). If
$u^{t}$ were a \emph{known} function, the r.h.s of \eqref{(6.1)} would
be its usual asymptotic expansion, which would determine it with any
degree of precision. If, however, nothing is known about the l.h.s. -
as is the case with the Casimir effect - the r.h.s. of \eqref{(6.1)}
must be taken as the \emph{definition} of $u^{t}$, and then Wightman's
Observation A applies. Indeed, by this definition, $u^{t}=+\infty$,
since the $O(d^{-4})$ term in \eqref{(6.1)} indicates that one is
supposed to sum the series, which, however, diverges, whatever
(nonzero) value of the small parameter is filled in. Due to
\eqref{(4.17)} and \eqref{(5.1)}, every finite approximation to
$u^{t}$ is independent of the cutoff $g$, but the rest diverges as
$N \to \infty$. Taking a sufficiently large number of terms, even the
sign of $u^{t}$ eventually changes from negative to positive, in
analogy to \eqref{(4.14)} and \eqref{(4.15)}.

It is to be remarked that the above-mentioned ``residual divergence'' is not removed by any process of renormalization, and is, in this respect, quite analogous to the situation in perturbative qed \eqref{(1)}, which refers to the \emph{renormalized} perturbation series (for the gyromagnetic ratio of the electron). It is present in all approaches which use the Euler-MacLaurin series. This does not, of course, mean that these approaches are ``wrong'': it means that they are not \emph{mathematically precise}, the issue being one of striving towards a higher level of understanding.

We should like to expand slightly on this important issue, because it touches on the philosophy of science. As Jaffe observes in \cite{Jaf}, lesson III, p.7: ``Arthur (Wightman) insisted: A great physical theory is not mature until it has been put in precise mathematical form''. As discussed in section 2, perturbative qed is also, in a similar way, not mathematically precise, in spite of remaining one of the greatest successes of physics, but, as Lieb observes in \cite{Lieb}, ``it is as much an enigma as it is a success''.  One important point in this connection is that perturbation theory provides a wrong picture of the photon cloud which surrounds an electron, see \cite{Lieb} and references given there, as well as \cite{LLrel} and \cite{JaWre2}. 

What can be said, in analogy, about the Casimir effect? It is, certainly one of the very few \emph{nonperturbative} (in the sense of being presumably exact) effects of qed. As remarked in \cite{Sp}, p. 170, together with blackbody radiation, it provides the most direct (experimental) evidence for the quantum nature of the Maxwell field (to which one might add the phenomenon of spontaneous emission, see \cite{JJS}). On the theoretical side, there are strong \emph{conceptual} arguments which require that the electromagnetic field be quantized \cite{BR}. In this same paper, Bohr and Rosenfeld point out that the square of the fields at a single point, such as in the first expression for the vacuum energy $\Sigma_{\Lambda}$, are ill-defined. In fact, as discussed in \cite{JJS}, p. 33, what we measure by a test body is the field strength averaged over some small region about a point: fields are what is termed operator-valued distributions, a notion which is basic to the axiomatic (or general) theory of quantized fields, according to which only the so-called Wick dots $:(\vec{E}(\vec{x}))^{2}:$ exist \cite{Wig3}. Building on this notion, a few different approaches to the Casimir effect, which do not use the Euler-MacLaurin series (see \cite{SW}, \cite{BrMac}, \cite{Niek}), arrive at the same result independently. Of particular interest is the paper \cite{BrMac}, which uses the image method in a field theoretic context and arrives at \eqref{(5.14)} by summing a convergent series. These different conceptual formulations are free of infinities, even of these mentioned in connection with \eqref{(6.1)}, but they are rather special, in contrast to Tao's formulation, which is quite general: there, the would-be rest in \eqref{(6.1)} disappears in the limit $N \to \infty$ (see Theorem 3.1).

In conclusion, the previously mentioned references, as well as Theorem ~\ref{thm:2}, are mathematically precise statements of the Casimir effect. In different ways, they introduce a ``smoothing'', which has its roots in the previously discussed singular nature of the quantum fields. This ``smoothing'' is familiar from distribution theory, which is one of the basic mathematical pillars of classical mathematical physics \cite{Sch}. Coming back to sequences and series of functions, as in section 2, consider the sequence of infinitely differentiable functions $f_{j}(x) \equiv \sin(jx), j=1.2, \cdots $ It certainly has no limit in the sense of functions, but let $\phi$ be an infinitely differentiable function, equal to zero outside the set $[-\pi,\pi]$. A partial integration shows that
$$
(T_{j},\phi) \equiv \int_{-\pi}^{\pi} \sin(jx) \phi(x) dx = \frac{1}{j} \int_{-\pi}^{\pi} \cos(jx) \phi^{'}(x) dx
$$
where the prime indicates differentiation. Thus, $\lim_{j\to \infty} (T_{j},\phi)= 0$. Similarly (see, e.g., \cite{BB}, Chap. 3, p.36), it may be shown that the series
\begin{equation}
\delta_{j}(x) \equiv \frac{1}{2\pi}\sum_{k=-j}^{k=j} \exp(ikx) = \frac{\sin((j+1/2)x)}{\sin(x/2)}
\label{(6.2)}
\end{equation}
is such that
\begin{equation}
\label{(6.3)}
\lim_{j\to \infty} (\delta_{j},\phi) = \lim_{j\to \infty}\int_{-\pi}^{\pi} dx \delta_{j}(x) \phi(x) = \phi(0)
\end{equation}
that is, ``$\delta_{j}$'' is a ``delta-sequence'': the r.h.s of \eqref{(6.2)} becomes more and more concentrated around the point $x=0$ as $j$ grows large, but the ``distributional limit'' \eqref{(6.3)} does exist. That is, a smoothing around the singular point $x=0$ enables the limit to exist. In analogy, with \eqref{(5.15)} and \eqref{(5.17)},
\begin{equation}
\lim_{N \to \infty} u_{N}^{t} = \mbox{ const. } \eta(0+)
\label{(6.4)}
\end{equation}
This means that a smoothing of the ``steps'' at each integer $N$ also enables the limit $N \to \infty$ in Theorem 3.1 to exist. Since the derivative of the step function is a ``delta function'' in the sense of distributions (see, e.g. \cite{Sch}, p. 82), the two notions are related.

We now come to more general nonperturbative approaches, in which the singular nature of quantum fields also plays a major role.

\section{General aspects of nonperturbative quantum fields: Wightman axioms for interacting quantum fields, dressed particles in a charged sector and unstable particles}

We shall try to keep with the general objective of remaining at the theoretical physicist's mathematical level (\cite{Barton1}, \cite{Weinb1}). The theory will be defined by its vacuum expectation values (VEV) or \emph{n-point Wightman functions} of a (for simplicity) scalar field $A(x)$: $W_{n}(x_{1}, \cdots, x_{n}) = (\Omega, A(x_{1}), \cdots, A(x_{n}) \Omega)$, where $\Omega$ denotes the vacuum. Because of the previously mentioned singular nature of the quantum fields $A(x)$, the $A(x)$ are not (operator-valued) functions of the space-time variables $x=(x_{0},\vec{x})$, but rather \emph{functionals}, denoted by 
$A(f)$, which may be heuristically pictured as ``smeared'' objects $\int dx f(x) A(x)$, with $f$ smooth, fast decreasing at infinity functions, taken to belong to the Schwartz space \cite{Sch}: one speaks of ``operator-valued tempered distributions'' on the Schwartz space ${\cal S}$.

In \cite{StreWight}, pp.~107-110, it is shown that if the (n-point) \emph{Wightman functions} 
satisfy 
\begin{itemize}
\item [$a.)$] the relativistic transformation law; 
\item [$b.)$] the spectral condition; 
\item [$c.)$] hermiticity; 
\item [$d.)$] local commutativity; 
\item [$e.)$] positive-definiteness, 
\end{itemize}
then they are the \emph{vacuum expectation values of a field theory} satisfying 
the so-called Wightman axioms, except, eventually,  the uniqueness of 
the vacuum  state. We refer to \cite{StreWight}  for 
an account of Wightman theory; see also \cite{Wig3}.

We shall \emph{assume} $a.)-e.)$ for the 
$n$-point functions of the observable  fields, with, in addition, the following requirement
\begin{itemize}
\item [$f.)$] interacting fields are assumed to satisfy the singularity hypothesis (the forthcoming  
Definition).
\end{itemize}

\subsection{The K\"all\'en-Lehmann representation}

In a non-perturbative framework, there exists a theory of renormalization 
of masses and fields, which ``has nothing directly to do with the presence of 
infinities'' (\cite{Weinb1}, p.~441, Sect.~10.3). We adopt a related proposal, which we formulate 
as previously mentioned, for simplicity, for a theory of a self-interacting scalar field $A$ of mass $m$ satisfying 
the Wightman axioms, assumed to be an operator-valued 
tempered distribution on the Schwartz space ${\cal S}$. 

\bigskip
The following  result is the spectral representation of the 
two-point function $W_{2}$ (\cite{Weinb1}, p. 457):
\begin{equation}
		W_{2}^{m}(x-y)= \langle \Omega,A(x)A(y)\Omega \rangle
		= \frac{1}{i} \int_{0}^{\infty} {\rm d}\rho(m_\circ^{2}) \; 
		\Delta_{+}^{m_\circ}(x-y) \; , 
		\label{(7.1)}
\end{equation}
where $\Omega$ denotes the vacuum vector, $x=(x_{0},\vec{x})$, and
	\begin{align}
		\Delta_{+}^{m_\circ}(x) =\frac{i}{2(2\pi)^{3}} 
		\int_{\mathbb{R}^{3}} {\rm d}^{3}\vec{k} \; 
		\frac{ {\rm e}^{-ix_{0}\sqrt{m_\circ^{2}+\vec{k}^{2}}
		+i\vec{x}\cdot\vec{k}}}{\sqrt{m_\circ^{2}+\vec{k}^{2}}} 
		\label{(7.2)}
	\end{align}
is the two-point function of the free scalar field of mass $m_\circ$. It is further assumed that
	\begin{equation}
		\langle \Omega,A(f)\Omega \rangle = 0 \qquad \forall  f \in {\cal S} \; . 
		\label{(7.3)}
        \end{equation}
The spectral measure $\rho$ may be further decomposed in the sum of a discrete and a continuous part
\begin{equation}
		{\rm d}\rho(m_\circ^{2}) = Z\delta(m_\circ^{2}-m^{2}) + d\sigma(m_\circ^{2}) \; , 
		\label{(7.4)}
\end{equation}
where $\sigma$ is a continuous measure and
	\begin{equation}
		0 \le Z < \infty
		\label{(7.5)}
\end{equation}

Of course, if $Z=0$ in \eqref{(7.4)}, there is no discrete component of mass $m$ in the total 
mass spectrum of the theory. The positivity of $Z$ is due to the positive-definiteness conditions e.) 
(or the positive-definite Hilbert space metric). See \cite{Barton1}, p. 50, or (\cite{Weinb1}, p.~461, Equ.~(10.7.20)), 
where it is shown that $Z$  is the (non-perturbative) wave function renormalization constant.

\subsection{The Singularity Hypothesis}

One of the most important features of relativistic quantum field theory 
is the behaviour of the theory at large momenta (or large energies). Renormalization group theory \cite{Weinb2}
has contributed a significant lore to this issue (even if none of it has been made entirely rigorous): it strongly suggests
that the light-cone singularity of the two-point functions
of interacting theories is stronger than that of a free theory: this is expected even in asymptotically free
quantum chromodynamics, where the critical exponents are anomalous. We refer to this as the ``singularity hypothesis'',
which will be precisely stated in the next section. This hypothesis is also verified in each order of perturbation theory
(if the interaction density has engineering dimension larger than 2).

\subsubsection{Steinmann Scaling Degree and a theorem}

In order to formulate the singularity hypothesis in rigorous terms, the 
Steinmann scaling degree $sd$ of a distribution \cite{Steinmann} is a natural concept: for 
a distribution $u \in {\cal S}^{'}(\mathbb{R}^{n})$, let $u_{\lambda}$ denote 
the ``scaled distribution'', defined by 
	\[
		u_{\lambda}(f) \equiv u(f(\lambda^{-1} \cdot)) \; . 
	\]
As $\lambda \to 0$, we expect that $u_{\lambda} \approx \lambda^{-\omega}$ for 
some $\omega$, the ``degree of singularity'' of the distribution $u$. Hence, we set
	\begin{equation}
		sd(u) \equiv \inf \, \bigl\{\omega \in \mathbb{R} \mid \lim_{\lambda \to 0} 
		\lambda^{\omega} u_{\lambda} = 0 \bigr\} \; , 
		\label{(7.6)}
	\end{equation}
with the proviso that if there is no $u$ satisfying the limiting condition above, 
we set $sd(u) = \infty$. For the free scalar field of mass $m \ge 0$, it is 
straightforward to show from the explicit form of the two-point function in terms of 
modified Bessel functions that
	\begin{equation}
		sd(\Delta_{+}) = 2 \; . 
		\label{(7.7)}
	\end{equation}
(see, e.g., \cite{Barton1}, (5.15)). In \eqref{(7.7)}, and the forthcoming equations, we omit the mass superscript. 
  
\begin{definition}
\label{def:1} 
We say that the \emph{singularity hypothesis} holds for an interacting scalar field if 
	\begin{equation}
		sd(W_{+}) > 2 \; . 
		\label{(7.8)}
	\end{equation}
\end{definition}

In \cite{JaWre1}, it was proved that:

\begin{theorem}
\label{thm:3} If the total spectral mass is finite, \emph{i.e.},
	\begin{equation}
		\int_{0}^{\infty} {\rm d}\rho(a^{2}) < \infty \; , 
		\label{(7.9)}
	\end{equation}
then
	\begin{equation}
		sd(W_{+}) \le 2 \; ; 
		\label{(7.10)}
	\end{equation}
\emph{i.e.}, the scaling degree of $W_{+}$ cannot be strictly greater than 
that of a free theory, and thus, by Definition~\ref{def:1}, the singularity 
hypothesis~\eqref{(7.8)} is \emph{not} satisfied.
\end{theorem}

\begin{corollary}
\label{cor:1}
The singularity hypothesis holds for an interacting scalar field only if 
$\int_{0}^{\infty} d\sigma(m_\circ^{2}) = \infty$. This necessary condition is independent
of the value of $0 \le Z < \infty$.
\end{corollary} 

The importance of the above theorem, 
and especially of its corollary
is that it provides a mathematical foundation for the forthcoming 
interpretation of the condition
        \begin{equation}
		Z=0 \; . 
		\label{(7.11)}
	\end{equation}
In this sense it is a complement to the foundations of quantum field theory. In order to 
understand why this is so, we have to make a brief interlude.

\subsubsection{The ETCR hypothesis and its consequences for the singularity hypothesis}

For the purposes of identification with Lagrangian field theory, one may equate the $A(.)$ of
\eqref{(7.1)} with the ``bare'' scalar field $\phi_{B}$ (\cite{Weinb1}, pg. 439), whereby
        \begin{equation}
		A = \sqrt{Z} A_{phys}
		\label{(7.12)}
	\end{equation}
under the condition 
        \begin{equation}   
                Z>0
                \label{(7.13)}
        \end{equation} 
Under the same condition \eqref{(7.13)}, the assumption of 
equal time commutation relations (ETCR) 
for the physical fields may be written (in the distributional sense)
		\begin{equation}
			\left[ \frac{\partial A_{phys}(x_{0},\vec{x}\,)}{\partial x_{0}},
			A_{phys}(x_{0},\vec{y}\,) \right] = -\frac{i}{Z} \; \delta(\vec{x}-\vec{y}\,) \; . 
			\label{(7.14)}
		\end{equation}
Together with \eqref{(7.1)} and \eqref{(7.14)}, one obtains (\cite{Barton1}, (9.19), \cite{Weinb1}, (10.7.18) (suitably modified by the factor
$\frac{1}{Z}$):

	\begin{equation}
		\frac{1}{Z} = \int_{0}^{\infty} {\rm d}\rho(m_\circ^{2}) \; . 
		\label{(7.15)}
	\end{equation} 

Formula \eqref{(7.15)} has been extensively used as a heuristic guide, 
even, for instance, by the great founders of axiomatic (or general) quantum 
field theory, Wightman and Haag. Indeed, 
in \cite{Wig4}, p.~201, it is observed that ``$\int_{0}^{\infty} {\rm d}\rho(m_\circ^{2}) = \infty$ 
is what is usually meant by the statement that the field-strength renormalization 
is infinite''. This follows from \eqref{(7.15)}, with ``field-strength renormalization'' interpreted 
as $\frac{1}{Z}$. The connection with the singularity hypothesis comes 
next (\cite{Wig4}, p.~201), with the observation that, by \eqref{(7.1)}, $W_{2}$ will 
have the same singularity, as $(x-y)^{2}=0$, as does $\Delta_{+}(x-y;m^{2})$. 
As for Haag, he remarks (\cite{Haag}, p.~55):``In the renormalized perturbation 
expansion one relates formally 
the true field $A_{phys}$ to the canonical field $A$ (our notation) which satisfies 
\eqref{(7.12)}, where $Z$ is a constant (in fact, zero). This means that the fields 
in an interacting theory are more singular objects than in the free theory, and we 
do not have the ETCR.'' 
Both assertions seems to substantiate the 
conjecture that $Z=0$ is expected to be a \emph{general} condition for interacting 
fields. In this connection, we have the following direct consequence of 
Corollary~\ref{cor:1}:

\begin{corollary}
\label{cor:2}
If \eqref{(7.15)} holds, only $Z=0$ is compatible with the singularity hypothesis.
\end{corollary}  

It follows from Corollary~\ref{cor:1} and Corollary~\ref{cor:2} that 
the two definitions of~$Z$, in \eqref{(7.4)}
and in \eqref{(7.15)} are \emph{not} equivalent. Since the ETCR, which 
implies \eqref{(7.15)}, is not
generally valid for interacting fields, as briefly reviewed in the 
forthcoming paragraph, we conclude
that the singularity hypothesis opens the possibility of the 
non-universal validity of \eqref{(7.11)}.

The hypothesis of ETCR has been in serious doubt for a long time, see, 
\emph{e.g.}, the remarks in \cite{StreWight}, p.~101. Its validity has been 
tested \cite{WFW} in a large class of models in two-dimensional space-time, where it was 
definitely proved not to hold in the case of the Thirring model \cite{Thirrm} for 
large coupling.  The only case found in \cite{WFW} in which the ETCR holds was
the Schwinger model, which is known \cite{LSwi} to be equivalent to a canonical theory
of a massive vector field! A nother interesting case of canonical behavior were the
canonical interacting quantum fields on two-dimensional de Sitter space \cite{JaMu}, where, 
however, the canonical character was intrinsically due to the geometry of de Sitter space.

Unfortunately, however, the ETCR (for interacting fields!) is still assumed to hold, 
without comment, in standard treatises (\cite{Weinb1}, p. 460; \cite{Barton1}, p.47).

Although, when $0<Z< \infty$, $Z$ is interpreted as
the non-perturbative field strength renormalization, relating ``bare'' fields to physical
fields, as in \eqref{(7.12)}, the remaining case \eqref{(7.11)} remains to be understood.
As stated in (\cite{Weinb1}, pg. 461), ``the limit $Z=0$ has an interesting interpretation
as a condition for a particle to be composite rather than elementary''. This brings us to
our next topic.

\section{A proposal for the meaning of the condition $Z=0$: the presence of massless and unstable particles} 

In the presence of massless 
particles (photons), Buchholz \cite{Buch} used Gauss' law to show that the discrete spectrum 
of the mass operator
	\begin{equation}
		P_{\sigma} P^{\sigma} = M^{2} = P_{0}^{2}-\vec{P}^{2}
		\label{(7.16)}
	\end{equation}
is empty. Above, $P^{0}$ is the generator of time translations in the physical 
representation, \emph{i.e.}, the physical hamiltonian $H$, and $\vec{P}$ is the physical 
momentum. This fact is interpreted as a confirmation of the phenomenon that 
particles carrying an electric charge are accompanied by clouds of soft photons.

Buchholz formulates adequate assumptions which must be valid in order that one may 
determine the electric charge of a physical state $\Phi$ with the help of Gauss' law 
		\begin{equation}
			\langle \Phi, j_{\mu} \Phi \rangle= \;
                        \langle \Phi, \partial^{\nu}F_{\nu,\mu} \Phi \rangle \; . 
			\label{(7.17)}
		\end{equation}
$F_{\nu,\mu}$ denotes the electromagnetic field observable, and \eqref{(7.17)} is assumed to hold in the sense of distributions 
on ${\cal S}(\mathbb{R}^{n})$. 

When endeavouring to apply Buchholz's theorem to concrete models such as $qed_{1+3}$, 
problems similar to those occurring in connection with the charge superselection 
rule \cite{StrWight} arise. The most obvious one is that Gauss' law 
\eqref{(7.17)} is only expected to be 
valid (as an operator equation in the distributional sense) in non-covariant 
gauges, the Coulomb gauge in the case of $qed_{1+3}$,  but not in 
covariant gauges \cite{StrWight}. If we adopt the present framework, our option is 
to use the Coulomb gauge and to define the theory in terms of  the 
$n$-point Wightman functions of observable  fields, 
\emph{i.e.}, gauge-invariant fields, thus maintaining 
Hilbert-space positivity. The hypotheses of Buchholz's theorem are then 
in consonance with the requirements of Wightman's theory~\cite{StreWight}, 
and are applicable to $qed_{1+3}$. 

In a charged electron sector, denoting spinor indices by $\alpha,\beta$, we have
	\begin{align}
		S_{\alpha,\beta}^{+}(x-y) 
		& = \langle \Omega, \Psi_{\alpha}(x)\bar{\Psi}_{\beta}(y) \Omega \rangle
		\label{(7.18.1)} \\
		& = \int_{0}^{\infty} {\rm d}\rho_{1}(m_\circ^{2}) S_{\alpha,\beta}^{+}(x-y;m_\circ^{2}) 
		+ \delta_{\alpha,\beta} \int_{0}^{\infty} {\rm d} \rho_{2}(m_\circ^{2}) \Delta^{+}(x-y;m_\circ^{2}) \, , 
		\nonumber
	\end{align}
with ${\rm d} \rho_{ph}, {\rm d} \rho_{1}, {\rm d} \rho_{2}$ positive,  
measures, and $\rho_{1}$ satisfying certain bounds with respect 
to $\rho_{2}$ (\cite{Lehmann}, p.~350).
	\begin{equation}
		{\rm d} \rho_{1}(m_\circ^{2}) = Z_{2}\delta(m_\circ^{2}-m_{e}^{2}) + {\rm d} \sigma_{1}(m_\circ^{2}) \; , 
		\label{(7.18.2)}
	\end{equation}
with $m_{e}$ the renormalized electron mass, according to conventional notation. 
Recalling ~\eqref{(7.18.1)}, we have the following immediate
corollary of Buchholz's theorem: 

\begin{corollary}
\label{cor:3} 
For $qed_{1+3}$ in the Coulomb gauge, assuming it exists in the sense 
of the framework of this section
and satisfies the assumptions of Buchholz's theorem, the following condition holds:
and
	\begin{equation}
		Z_{2} = 0 \; . 
		\label{(7.19)}
	\end{equation}
\end{corollary}
  
Above, $\Psi$ denote observable fermion fields, which we assume to exist as a generalization of those 
constructed by Lowenstein and Swieca \cite{LSwi} in
$qed_{1+1}$, see also \cite{Steinmann1} for a similar attempt in perturbative $qed_{1+3}$. In the words of Lieb and Loss~\cite{LLrel}),
who were the first to observe this phenomenon in a relativistic model of qed,
 ``the electron Hilbert space is linked to the photon Hilbert space in an inextricable way''. Thereby, in this way, 
``dressed photons'' and ``dressed electrons'' arise as new entities. 

In order to stress the practical effects of the confusion of the two definitions of $Z$ in the
literature, when the ETCR is assumed (including \cite{Weinb1}), it should be remarked that the first
reference given by Weinberg on a model for unstable particles \cite{HouJou} assumes (what ammounts to) $m_{C}\le 2m$
instead of \eqref{(7.20)} (that is, we are in the stability range!), and finally obtains $Z=Z_{C}=0$!
This is solely due to the ``double definition'' of the non-perturbative wave function renormalization
constant!

We now come back to Weinberg's suggestion that the condition $Z=0$ describes unstable particles. 

Turning to scalar fields for simplicity, we consider the case of a scalar particle $C$, 
of mass $m_{C}$, which may decay into a set of two (for simplicity) stable particles, each 
of mass $m$. We have energy conservation in the rest frame of $C$,  \emph{i.e.}, 
	\[
		m_{C}=\sum_{i=1}^{2} \sqrt{{\vec{q}_{i}\,}^{2}+m_{i}^{2}} 
		\ge \sum_{i=1}^{2} m_{i} \; , 
	\]
with $m_{i}=m,i=1,2$, and $\vec{q}_{i}$  the 
momenta of the two particles in the rest frame of $C$:
	\begin{equation}
		m_{C} > 2m  \; . 
		\label{(7.20)}
        \end{equation}
In order to check that $Z=Z_{C}=0$ when \eqref{(7.20)} holds, while $0<Z<\infty$ is valid
in the stable case $m_{C} < 2m$, in a model, we are beset with the difficulty to obtain information 
on the two-point function. 

There exists a quantum model of Lee type of a composite (unstable) particle, 
satisfying \eqref{(7.20)}, where \eqref{(7.11)}
was indeed found, that of Araki et al.~\cite{AMKG}. Unfortunately, however, the (heuristic) results
in \cite{AMKG} have one major defect: their model contains ``ghosts''. A very good review of the
existent (nonrigorous) results on unstable particles is the article by Landsman \cite{Landsman}, to which
we refer for further references and hints on the intuition behind the criterion \eqref{(7.11)}.

In the next section we come back to a set of models for atomic resonances and particles, which 
might support the suggested picture of quantum field theory in terms of ``dressed'' and unstable particles,
and in the last section we discuss the crucial conceptual issues and difficulties associated with this program,
comparing it with alternative approaches. 

\section{Models for atomic resonances, unstable and ``dressed'' particles: what distinguishes quantum field theory from many-body systems?}

In \cite{Wre1} the model below - the Lee-Friedrichs model of atomic resonances - was revisited. Its Hamiltonian may be written
       \begin{equation}
               H = H_{0} + H_{I}
               \label{(8.1)}
       \end{equation}
with
       \begin{equation}
              H_{0} = E_{0} \frac{\mathbf{1}+\sigma_{z}}{2}\otimes \mathbf{1} + \mathbf{1}\otimes\int d^{3}k |k|a^{\dag}(k)a(k)
              \label{(8.2)}
       \end{equation}
and
       \begin{equation}
              H_{I}^{1} = \beta[\sigma_{-}\otimes a^{\dag}(g)+\sigma_{+}\otimes a(g)]
              \label{(8.3.1)}
       \end{equation}
The operators act on the Hilbert space
       \begin{equation}
             {\cal H} \equiv {\cal C}^{2} \otimes {\cal F}
             \label{(8.4)}
       \end{equation}
where ${\cal F}$ denotes symmetric (Boson) Fock space on $L^{2}(\mathbf{R}^{3})$ (see, e.g., \cite{MaRo}), which
describes the photons. We shall denote by $(\cdot,\cdot)$ the scalar product in ${\cal H}$. 
Formally, $a(g) \equiv \int d^{3}k g(k) a(k)$, and $k$ denotes a three-dimensional vector.
The $\dag$ denotes adjoint, $\sigma_{\pm}=\frac{\sigma_{x}\pm \sigma_{y}}{2}$, and $\sigma_{x,y,z}$ are the usual
Pauli matrices. The operator
       \begin{equation}
            N = \frac{\mathbf{1}+\sigma_{z}}{2}\otimes \mathbf{1}+ \mathbf{1} \otimes \int d^{3}k a^{\dag}(k)a(k)
            \label{(8.5)}
       \end{equation}
commutes with $H$. We write
       \begin{equation}
            N = \sum_{l=0}^{\infty} l P_{l}
            \label{(8.6)}
       \end{equation}
and introduce the notation
       \begin{equation}
           H_{l} \equiv P_{l} H P_{l}
           \label{(8.7)}
       \end{equation}
$H_{l}$ is the restriction of $H$ to the subspace $P_{l} {\cal H}$. 
Let in addition
        \begin{equation}
             E_{0} > \beta^{2} \int d^{3}k |k|^{-2} |g(k)|^{2}
             \label{(8.8)}
        \end{equation}
Then the one-dimensional subspace $P_{0}{\cal H}$  consists
of the ground state vector
       \begin{equation}
           \Phi_{0} \equiv |-) \otimes |\Omega)
           \label{(8.9)}
       \end{equation}
with energy zero, where
       \begin{equation}
           \sigma_{z} |\pm) = \pm |\pm)
           \label{(8.10)}
       \end{equation}
denote the upper $|+)$ and lower $|-)$ atomic levels, and $|\Omega)$ denotes the zero-photon state in ${\cal F}$. Note that
$\Phi_{0}$ is also eigenstate of the free Hamiltonian $H_{0}$, with energy zero, and we say therefore that the model
has a persistent zero particle state. 

We shall refer to the model described by \eqref{(8.1)} as \emph{Model 1}. Replace, now, in Model 1, $H_{I}^{1}$ by
       \begin{equation}
              H_{I}^{2} = \beta[(\sigma_{-}+\sigma_{+})\otimes (a^{\dag}(g)+ a(g))]
              \label{(8.3.2)}
       \end{equation}
We shall refer to the ensuing model as \emph{Model 2}. Let, now, $N,V$ satisfy \emph{anticommutation relations}, i.e.,
$$
\{N,N^{+}\}=\{V,V^{+}\}=1 
$$
together with all other anticommutators equal to zero, i.e. $\{N,V^{+}\}=0$, etc.
By means of Schwinger's representation
\begin{eqnarray*}
\sigma_{+}=V^{+}N \\
\sigma_{-}=N^{+}V \\
\sigma_{3}= \frac{1}{2}(V^{+}V-N^{+}N)
\end{eqnarray*}
with the further correspondences
\begin{eqnarray*}
|+) = V^{+}|0) \\
|-) = N^{+}|0)
\end{eqnarray*}
where $|0)$ denotes the fermion no-particle state, Model 1 becomes the usual Lee model for particles, and Model 2 a
``refined'' Lee model for particles; we refer to them as \emph{Model 3} and \emph{Model 4} respectively.
There is one big difference, however, between the atomic resonance case and the
particle case. In the former case, the function $g$ in \eqref{(8.3.1)} or \eqref{(8.3.2)} is square-integrable, and,
indeed, the physical dipole-moment matrix elements provide natural cutoffs, so that neither infrared nor ultraviolet
problems arise. In the latter case, however, one aims at the pointwise limit $g(k) \to 1$, which should lead to
an euclidean-invariant theory. This step is delicate and requires mass and wave function renormalizations. To the
particle versions, Model 3 and Model 4, assuming they are well-defined, the forthcoming framework is applicable.

Let a theory of a scalar field of mass $m>0$ be invariant under the euclidean group, that is, the group of translations and 
rotations of euclidean space $\vec{x} \to R\vec{x}+\vec{a}$, where $R$ denotes a rotation. By Haag's theorem (\cite{Wig4}, p. 249),
in a euclidean field theory which uses the Fock representation, the no-particle state $\Psi_{0}$ is euclidean invariant, i.e.,
\begin{equation}
U(\vec{a},R)\Psi_{0}=\Psi_{0}
\label{(8.11)}
\end{equation}
We have the
\begin{theorem}
\label{thm:4}
Let the Hamiltonian be of the form 
\begin{equation}
H=\int {\cal H}(\vec{x})d\vec{x}
\label{(8.12)}
\end{equation}
where ${\cal H}(\vec{x})$ satisfies
\begin{equation}
U(\vec{a},\mathbf{1}) {\cal H}(\vec{x} U(\vec{a},\mathbf{1})^{-1} = {\cal H}(\vec{x}+\vec{a}) 
\label{(8.13)}
\end{equation}
Then, if $\Psi$ is any state invariant under $U(\vec{a},\mathbf{1})$, i.e.,
\begin{equation}
U(\vec{a},\mathbf{1}) \Psi = \Psi
\label{(8.14)}
\end{equation}
then $\Psi$ belongs to the domain of $H$ only if $H \Psi =0$.
\end{theorem}

\begin{proof}
We have that $(\Psi, {\cal H}(\vec{x}){\cal H}(\vec{y})\Psi)$ depends only on $\vec{x}-\vec{y}$, so that
\begin{equation}
||H\Psi||^{2} = \int\int d\vec{x} d\vec{y}(\Psi, {\cal H}(\vec{x}){\cal H}(\vec{y})\Psi) = 0 \mbox{ or } \infty
\label{(8.15)}
\end{equation}
according to whether $H\Psi = 0$ or $||H\Psi|| = \infty$, that is, $H$ must annihilate any translation-invariant state to which it
is applicable.
\end{proof}

Choosing $\Psi = \Psi_{0}$ in the Fock representation, it follows from theorem ~\ref{thm:4} together with the consequence of Haag's theorem 
\eqref{(8.11)} that $H$
can be applied to the no-particle state only if it annihilates it. The above was extracted from \cite{Wig4}, p. 250, and to make it
entirely rigorous \eqref{(8.12)} should be written as a limit of the ``smearing'' of ${\cal H}$ with smooth functions. This result
may be applied to Model 3 (assuming the renormalizations performed such that the pointwise limit can be taken, leading to an euclidean
invariant quantum field theory), with the no-particle state identified to the state given by $\Psi \equiv |\Omega)\otimes |-)$, in
and $|\Omega)$ the no-particle photon state: $H \Psi = 0$. We say that there is
no \emph{vacuum polarization}. On the other hand, for Model 4, $\Psi$ does not belong to the domain of $H$ due to the term  $\sigma_{+}a^{\dag}(g)$
in \eqref{(8.3.2)}, which, in terms of fermion operators equals $V^{+}N a^{\dag}(g)$, and $V^{+}N |-)= |+)$. In the case there is vacuum polarization,
theorem ~\ref{thm:4} implies that there exists an ``infinite energy barrier'' between the Fock no-particle state and the true vacuum: non-Fock
representations are required, that is, the ``physical'' Hilbert space is not unitarily equivalent to Fock space.   

All this been said, Model 3 turns out, very unexpectedly, to be afflicted by ``ghosts'', i.e., states of negative norm (see \cite{Barton1}, Chap. 12).
We say unexpectedly, because the occurrence of ``ghosts'' in relativistic quantum field theory is known to be a consequence of the use of
manifestly covariant gauges (see \cite{Weinb2}), and the Lee model is not relativistically invariant. Note, however, that in the Schwinger representation
the ``fermions'' $V$ and $N$ are two states of an \emph{infinitely heavy} (spinless) fermion, i.e., \emph{there is no recoil}! Of course, this is a highly unphysical assumption. One may consider, however, the model with recoil, describing the interaction between the photon field and \emph{two} particles $V$ and $N$ , with energies $E_{V}(p)= (M^{2}+p^{2})^{1/2}$ and $E_{N}(p)=(m^{2}+p^{2})^{1/2}$, and interaction energy 
$H_{I}=\int d^{3}p d^{3}k g(p,k)[V^{\dag}(p)N(p-k)a(k)+ \mbox{ h.c. }]\lambda$, with $\lambda$ proportional to the charge, and
$g(p,k) \equiv f(p,k)[8E_{V}(p)E_{N}(p-k)\omega(k)]^{-1/2}$, $\omega(k)=(\mu^{2}+k^{2})^{1/2}$, the latter representing the (eventually massive) ``photon'' energy. This model is well-defined and free of ``ghosts'' in the pointwise limit $f \to 1$, the latter taken in a careful way, according to a renormalization prescription. This was proved by Yndurain \cite{Ynd} in a seldom cited, but very important paper. Thus, the pathologies associated to the original Lee model just have to do with neglecting recoil! 

We henceforth refer to the Yndurain versions of models 3 and 4 as Model 3Y and Model 4Y.

Model 2 (for zero temperature, as a model of atomic resonances) may be the simplest prototype of a model with vacuum polarization. Model 3Y should be suitable to study the criterion $Z=0$ for unstable particles. The subtlest point in this connection is the fact that the term proportional to the delta measure in \eqref{(7.4)} has coefficient $Z$, and in its absence, due to the condition $Z=0$, the (renormalized) mass seems to remain undetermined. One must, therefore, be able to determine the renormalized mass uniquely from \emph{alternative} general conditions on the Hamiltonian, such as the requirement that it be self-adjoint and bounded below.

We close this section with some remarks of what distinguishes quantum field theory from many-body systems. The latter are characterized by an interaction
Hamiltonian $\int d^{3}x d^{3}y \sum_{\sigma,\sigma^{'}}V(x-y)\Psi^{\dag}(x,\sigma)\Psi^{\dag}(y,\sigma^{'})\Psi(y,\sigma^{'})\Psi(x,\sigma)$, where $V$ denotes the interaction potential, $\sigma$ and $\sigma^{'}$ spin indices and $\Psi$ are boson or fermion quantized operators (see, e.g., \cite{MaRo}, p. 110).
The interaction term, and consequently the whole Hamiltonian, annihilates therefore the no-particle state (Fock vacuum). In contrast, quantum field theories  display in general vacuum polarization, as typified by Model 4Y, which is an approximation of the usual trilinear coupling terms which occur in qed and also in quantum chromodynamics (qcd), as a consequence of relativistic invariance and the gauge principle (see \cite{Weinb1}, \cite{Weinb2}). Coupled with Haag's theorem and theorem ~\ref{thm:4}, this implies non-Fock representations, as previously discussed. This is a further manifestation of the singular nature of quantum fields. Of course, the above reference to many-body systems concerns a finite number of particles. For nonzero density, i.e., an infinite number of degrees of freedom, non Fock representations arise even in the case of free systems, by a well-known mechanism (\cite{Sewell},section 2.3). What we wished to emphasize is that in the case of quantum field theories, such representations arise due to a particular reason, namely, vacuum polarization.

\section{Conclusions}

In this review, we focused on the foundations of quantum field theory, which is still believed to be the most fundamental theory, describing in principle all phenomena observed in atomic and particle physics. Unlike quantum mechanics, however, its foundations are still not cleared up. We attempted to describe how some novel approaches lead to a unified picture, in~spite of the fact that several difficult open problems remain. The~solution of some of them seems to be nontrivial, but~feasible. Many other problems demand new ideas, however.

The first issues concern \emph{relativistic} quantum field theories, such as qed or qcd. They have been discussed in Sections 3 and and 4, after a pedagogic review of asymptotic and divergent series in Section 2.

The fact that divergent series, even if asymptotic, do not define a theory mathematically, was emphasized in the case of qed; see, in particular, \eqref{(10)}--\eqref{(12)}. From~the latter, we see that even the assumption that renormalized perturbation series, such as \eqref{(1)}, is the asymptotic series of an unknown function, which is probably true and explains the dazzling success of perturbative qed, cannot hold for qcd, where a rough dimensionless measure of the coupling constant is of order $1-10$. The~improved perturbation theory deriving from renormalization group arguments and asymptotic freedom in qcd (Reference \cite{Weinb2}, Ch. 18.7) has no comparable experimental~consequences. 

We, therefore, turned our attention to \emph{nonperturbative} phenomena, one of the very few (in qed) being the Casimir effect, for~simplicity in its simplest form, that of perfectly conducting parallel plates in the vacuum. We showed in Section 3 that Tao's method of smoothed sums~\cite{Tao} eliminates the ``residual infinity'' present in the Euler-maclaurin series for the energy density of the field (Theorem ~\ref{thm:2}). This smoothing of the series is a new method of accounting for the singular nature of quantum fields, which are not pointwise defined.

Going beyond specific models or effects, we reviewed in Section 4 a recent novel criterion to characterize interacting theories in the Wightman framework~\cite{JaWre1}. Here, a different aspect of the singular nature of quantum fields is touched upon: the nature of the singularity of the two-point Wightman function at small distances. Using the K\"{a}ll\'{e}n- Lehmann representation of the two-point function $W_{+}$ and the Steinmann scaling degree~\cite{Steinmann} $sd$, we propose to characterize interacting Wightman theories by the requirement $sd(W_{+})>2$. We were then able to prove a Theorem \ref{thm:3} which states that, in an interacting Wightman theory, the total spectral mass which occurs in the K\"{a}ll\'{e}n- Lehmann representation is infinite. This allows us, in turn, to state that, contrary to previous belief, the~\emph{nonperturbative} wave-function renormalization constant $Z$ is not universally equal to zero. We then interpret the condition $Z=0$ in section 5 in a two-fold way: it is either due to the nonexistence of a pure point part in the mass spectrum in charged sectors in theories, such as qed, due to infraparticles (electrons with their photon clouds), or~to the existence of unstable particles, as~conjectured by \mbox{Weinberg (Reference \cite{Weinb1}, p. 461.} We believe that this new approach is connected with the theories of ``dressed'' photons and electrons in Reference~\cite{JaWre2,LLrel}. 

In qcd, the~gluons being massless, a~similar phenomenon as discussed in the last paragraph should be expected: an interacting theory of ``dressed'' quarks and gluons. It has been suggested by Casher, Kogut and Susskind~\cite{CKS}, and Swieca~\cite{Swi} that massless $qed_{1+1}$ contains what is desired of a theory of quark confinement, in~the sense that under short distance probing the theory behaves as if it contained particles which do not manifest themselves as physical states: in that limit one recovers a theory of free (massless) electrons and photons. For~large distances, the electrons completely disappear from the picture, giving rise to massive photons, by~``Bosonization'' (an extreme form of the ``dressing'' phenomenon). In Reference~\cite{JaWre1}, we formulated a precise criterion for confinement which applies to $qed_{1+1}$ and which, generalized to qcd (under assumptions on the leading infrared singularity of the gluon propagator), yields a surprisingly realistic picture of confinement and asymptotic~freedom. 

We should like to mention a recent alternative approach, due to Buchholz and Fredenhagen (Reference \cite{BF} and references given there). There the assumed existence of a time-arrow yields a novel approach, centered on an interesting, non-commutative structure of dynamical algebras inspired by scattering theory. An~early reference which treats unstable particles in the spirit of open systems, associated to a dynamical group (the Poincar\'{e} semigroup), is the paper by Alicki, Fannes, and Verbeure~\cite{AFV}. Sewell's book~\cite{Sewell1} also deals with the emergent macrophysics originating from quantum mechanics, albeit not in the realm of quantum field~theory.

More directly aimed at gauge theories is the theory of string-localised fields, due to Mund, Schroer and Yngvason, which has some points of contact with the present approach, specifically the insistence on positivity, i.e., ``ghostless'' theories. The theory has recently progressed considerably through a deep analysis of Gauss' law in that connection, see \cite{MRS} and references given there to the previous literature. 

A theory of quantum fields in de Sitter space has been developed in recent years by Jaekel, Mund and Barata. It is particularly interesting because of its connection to thermal aspects of quantum fields, as early realized by Narnhofer, Peter and Thirring \cite{NaPeTh}. We refer to \cite{JaMu} for some recent references.  

In addition, an important perturbative algebraic quantum field theory has been developed in the last twenty years. Its recent application to the sine-Gordon model \cite{BFR} actually yielded a convergent series, which marks the success of this approach: we refer to \cite{BFR} to references to the previous extensive literature. An earlier treatment of integrable two-dimensional models, with several novel structural features, both conceptual and technical, is due to Lechner and collaborators, see the review by Alazzawi and Lechner \cite{AlLe} and references given there.

In spite of the fact that the convergence of the perturbation series for models such as the massive Thirring-Schwinger model ($qed_{1+1}$) is well-known since the early days of constructive quantum field theory (1976) \cite{FroSei}, it is important to emphasize that no singular point fields appear in the functional analytic approach of \cite{BFR}, so that it may be expected that this new method, invented by Fredenhagen, may be able to cope with physically relevant models in the future, such as gauge theories in four space-time dimensions.   

Concerning the latter, since the gluons are, as the photon, massless, we are proposing a picture of interacting gauge theories as composed of ``dressed'' quarks and gluons, or dressed electrons and photons, for qed. The latter has been suggested by Lieb and Loss \cite{LLrel} and Jaekel and myself \cite{JaWre2}, and we suggested further that this picture is characterized by the condition $Z=0$ \cite{JaWre1}. Unfortunately, even in $qed_{1+3}$ this involves a construction of fermion observable fields which has not even been achieved in perturbation theory \cite{Steinmann1}. This project is therefore very difficult for qed, and the self-interactions of the gluons in qcd make it even much more difficult there. We feel, however, that it contains a ``grain of truth'': in particular, Lieb and Loss provide convincing arguments to the effect that description by ``dressed'' photons is the more natural one in their relativistic model of qed, while we relate this phenomenon to the non-Fock structure of the representations (when the ultraviolet cutoff is removed). The latter appear in connection with a generalization of the transformations analized by Wightman and Schweber \cite{WigSch} which occur in $qed_{3+1}$ in the Coulomb gauge \cite{JaWre2}. Therefore, from our point of view, a first step would be a study of a simpler model displaying vacuum polarization. The simplest such is Model 2.

Although the Lee-Friedrichs model is very old (see \cite{Wre1} and references given there), it is one of the few approximations to relativistic quantum field models which preserve some basic physical features. This may be due to the fact that its space dimension is three, and that it retains the trilinear coupling familiar from relativistic theories. We proposed in Section 6 that the main feature distinguishing quantum field theory from many-body systems is \emph{vacuum polarization}. When the theories are euclidean invariant, Haag's theorem, together with theorem ~\ref{thm:4} provides a way to distinguish Model 3Y from Model 4Y (the letter Y refers to the important work of Yndurain \cite{Ynd}, who showed that the ``ghosts'' in the Lee model were solely due to neglecting recoil). We believe that Model 3Y is suitable to verify the prediction $Z=0$ for unstable particles. Model 4Y seems at present beyond any control, but understanding Model 2, which is relevant to atomic resonances, would be an enormous advance towards grasping the main features of vacuum polarization. It should be noted that Herbert Fr\"{o}hlich's electron-phonon theory \cite{HF} is a true quantum field theory with vacuum polarization, and the failure of being able to handle it may well be the ``penalty'' for the problems found in BCS theory \cite{Sewell1}. The study of Model 1 in \cite{Wre1} shows clearly that the field-theory interaction through emission and absorption of (virtual) particles is of \emph{qualitatively} different nature from potential theory: in particular, the regeneration of the unstable state from the decay products is a virtual quantum phenomenon analogous to tunneling in potential theory which is, however, not present in potential theory. In the case of atomic resonances, Model 1 even accounts even for subtle physical aspects, such as the fact that the natural line width of Lamb states is much smaller than the corresponding (Lamb) shift, essential for the observability of the latter \cite{Wre1}: this is an example of the physical features mentioned in the beginning of this paragraph.

\end{document}